\documentclass[a4paper]{amsart}

\usepackage{amsmath, amsthm, mathscinet}
\usepackage{tikz}
\usepackage{amssymb}
\usepackage{
	comment,
	enumitem, % enables a greater range of eumerate options
	float, % images option H
	wrapfig,
	subcaption }
\usepackage{microtype}
\usepackage[applemac]{inputenc}
\usepackage{graphicx}
\usepackage{color}
\usepackage{tikz} \usetikzlibrary{trees}
\usepackage{enumitem}
\setlist[enumerate,1]{label=(\roman*)}

\usetikzlibrary{matrix}
\makeatletter
\newcommand*\bigcdot{\mathpalette\bigcdot@{.6}}
\newcommand*\bigcdot@[2]{\mathbin{\vcenter{\hbox{\scalebox{#2}{$\m@th#1\bullet$}}}}}
\makeatother

\def\E{{\mathbb E}}
\def\F{{\mathcal F}}

\def\G{{\mathcal{G}}}

\def\I{{\cal I}}

\def\P{{\mathbb P}}
\def\Q{{\mathbb Q}}
\def\R{{\mathbb R}}

\def\I\nd{{\mathbb I}}

\renewcommand{\subset}{\subseteq}

\newtheorem{theorem}{Theorem}[section]

\newtheorem{lemma}[theorem]{Lemma}

\newtheorem{proposition}[theorem]{Proposition}

\newtheorem{definition}[theorem]{Definition}

\newtheorem{remark}[theorem]{Remark}
\newtheorem{example}[theorem]{Example}

\renewcommand{\epsilon}{\varepsilon}
\newcommand{\nd}{\text{nd}}

\renewcommand{\langle}{[}
\renewcommand{\rangle}{]}

\newcommand{\id}{\text{Id}}

\newcommand{\cpl}{\mathrm{Cpl}}
\newcommand{\cpla}{\mathrm{Cpl}_{\mathrm{C}}}
\newcommand{\cplba}{\mathrm{Cpl}_{\mathrm{BC}}}

\newcommand{\OB}{Ob{\l}{\'o}j}%{Ob{\l}{\'o}j} 

\numberwithin{equation}{section}
\title{Adapted Wasserstein Distances and Stability in Mathematical Finance
}
\author{J.\ Backhoff-Veraguas, D.\ Bartl, M.\ Beiglb\"ock, M.\ Eder}
\begin{document}
\begin{abstract}
Assume that an agent models a financial asset through a measure $\Q$ with the goal to price / hedge some derivative or optimize some expected utility. Even if the model $\Q$ is chosen  in the most skilful and sophisticated way, she is left with the possibility that $\Q$ does not provide an  \emph{exact} description of reality.  This leads us to the following question: will the hedge still be somewhat meaningful for models in the proximity of $\Q$?

%In fact, this fails dramatically 
%if the space of models is equipped with the usual L\'evy-Prokhorov metric or Wasserstein  distance (say).  

If we measure proximity with the usual Wasserstein distance (say), the answer is NO. Models which are similar w.r.t.\ Wasserstein distance may provide dramatically different %amounts of 
information on which to base a hedging strategy.

Remarkably, 
this can be overcome by considering a suitable \emph{adapted} version of the Wasserstein distance which takes the temporal structure of pricing models into account. This adapted Wasserstein distance is most closely related to the nested distance as pioneered by Pflug and Pichler \cite{Pf09,PfPi12,PfPi14}. It allows us to establish Lipschitz properties of hedging strategies for semimartingale models in discrete and continuous time. Notably, these abstract  results are sharp already for Brownian motion and European call options.
\end{abstract} 

\maketitle

%\end{keywords}
%} \subclass{ 91G80 \and 60G42 \and 60G44 \and 90C15 }

\section{Introduction}

\subsection{Outline}
Assume that a reference measure $\P$ is used to model the evolution of a financial asset $X$ with the purpose to hedge a financial claim or to maximize some expected utility. We do not expect that the model $\P$ captures reality in an absolutely accurate way. However, supposing that $\P$ is close enough to reality  (described  by a probability $\Q$)  we would still hope that a strategy which is developed for $\P$ leads to reasonable results. 
 
A main goal of this paper  is to establish this intuitive idea rigorously based on a new notion of \emph{adapted Wasserstein distance} $\mathcal{AW}_p$ between semimartingale measures. 
To fix ideas, we provide a first example of the results we are after.

\begin{theorem}
\label{thm:hedging.with.loss.linear.intro}
	Let $\P, \Q$ be continuous semimartingale models for the asset price process $X$,  and assume that $C(X)$ denotes an $L$-Lipschitz payoff of a (pathdependent) derivative $C$.  
		Assume that a predictable trading strategy $H=(H_t)_t$, $|H| \leq k$ and an initial endowment $m\in \R$ constitute a $\mathbb{P}$-superhedge of $C(X)$, i.e.
		\[ C(X) \leq m +(H\bigcdot X)_T, \quad \mathbb{P}\text{-almost surely}. \] 
						Then there is  a predictable $G$ s.t.\ $m,G$ constitute an ``almost'' $\mathbb{Q}$-superhedge:
	\begin{align}\label{eq almost superhedge intro}
	\mathbb{E}_\mathbb{Q}[(C(X)-m-(G\bigcdot X)_T)^+]\leq 6(k+L)\cdot \mathcal{AW}_1(\P, \Q).
	\end{align}
\end{theorem}
While the adapted Wasserstein distance will be defined in abstract terms (see \eqref{AWD}), it relates directly to the model parameters for `simple' models. In particular, if $\P, \Q$ are Brownian models with different volatilities, than the distance between these models is just the difference of these volatilities. Moreover, the  bound in \eqref{eq almost superhedge intro} (as well as further Lipschitz bounds given below) are already sharp in such a simple setting and for $C$ a European call option.

Below we will provide a number of results with similar flavour as Theorem \ref{thm:hedging.with.loss.linear.intro}. E.g.\ we will provide versions where the hedging error is controlled in terms of risk measures and we will show that a Lipschitz bound of the type  \eqref{eq almost superhedge intro} applies (with bigger constants) if the same trading strategy $H$ is applied in the model $\P$ as well as in the model $\Q$. Importantly, we establish that comparable results of Lipschitz continuity apply to utility maximization {and utility indifference pricing}.

We emphasize  that familiar concepts such as the L\'evy-Prokhorov metric or the usual Wasserstein distance do not appear suitable to derive results comparable to Theorem \ref{thm:hedging.with.loss.linear.intro}. {E.g.\ in the vicinity of financial meaningful models there are models with \emph{arbitrarily high arbitrage} even for bounded strategies; similar phenomena appear w.r.t.\ completeness / incompleteness.}  Instead we introduce an adapted Wasserstein distance $\mathcal{AW}_p$  which takes the temporal structure of semimartingale models into account. These distances are conceptually closely related to the nested distance as pioneered by Pflug and Pichler \cite{PfPi12,PfPi14,PfPi16}; see \cite{AcBaZa16,GlPfPi17,BiTa19} for first articles which link such a type of distance to finance. 
  We  describe these contributions more closely in Section \ref{sec:literature} below. 
  
  \subsection{Notation and adapted Wasserstein distances}

Throughout we let $$\Omega := \R^T\,\,\text{   or   }\,\,\Omega:= C(0,T).$$ 
The first setting shall be referred to as the discrete time case, and the second as the continuous time case.\footnote{Indeed the arguments in the discrete and the continuous case use the same set of ideas but the presentation is significantly less technical in the discrete case which was an important reason to include the discrete case  in the paper.} {In the first case we denote by $I=\{1,\dots,T\}$ the time-index set, and in the second $I=[0,T]$.} 
Throughout the article we will provide definitions and results without specifying which of the two cases we are referring to: This means that the definitions / results apply in both cases. Only occasionally will we consider one case specifically, and in this situation we will state this explicitly. 

We interpret $\Omega$ as the set of all possible evolutions (in time) of the 1-dimensional asset price. Importantly, mutatis mutandis,  all our results ({except Propositions \ref{lem synchron}, \ref{geoBM}  and Example \ref{ex det diff}}) remain true for multi-dimensional asset price processes (corresponding to $\Omega=(\R^d)^T$ / $\Omega =C([0,T], \R^d)$). We chose to go for the 1-dimensional version to simplify notation.

The mappings $X,Y\colon\Omega\to\Omega$ denote the canonical processes (i.e.~the identity map),  and we make the convention that on $\Omega\times\Omega$ the process $X$ denotes the first coordinate and $Y$ the second one.
The spaces $\Omega$ and $\Omega\times\Omega$ are endowed with the maximum-norm and the corresponding Borel-$\sigma$-field. In continuous time, the space $\Omega$ is endowed with the right-continuous filtration generated by $X$, in discrete time we use the plain filtration generated by $X$. {In any case we denote this filtration by $\mathbb F=(\F_t)_t$  and endow $\Omega\times\Omega$ with the product filtration $(\F_t\otimes \F_t)_t$. Given a $\sigma$-algebra $\G$ and a probability $\P$ on $\G$ we write $\G^\P$ for the $\P$-completion of $\G$.} 
%\textcolor{black}{For a probability measure $\mathbb{P}$ and a $\sigma$-field $\mathcal{G}$, we denote by $\mathcal{G}^{\mathbb{P}}$ the augmented $\sigma$-field.} 
The set $\cpl(\P, \Q)$ of couplings between probability measures $\P, \Q$ consists of all probability measures $\pi$ on $\Omega \times \Omega$ such that $X(\pi)=\P$ and $Y(\pi)= \Q$. A Monge coupling is a coupling that is of the form $\pi= (\id, T)(\P)$ for some Borel mapping $T:\Omega\to \Omega$ that transports $\P$ to $\Q$, i.e.\ satisfies $T(\P)=\Q$.  Given a metric $d$ on $\Omega$ and $p\geq 1$, the $p$-Wasserstein distance of $\P, \Q$ is 
\begin{align}\label{eq:UsualW}
\mathcal{W}_p(\P, \Q)=\inf\left\{\E_\pi[d(X,Y)^p]^{1/p}: \pi \in \cpl(\P, \Q)\right\}.
\end{align}
In many cases of practical interest the infimum in \eqref{eq:UsualW} remains unchanged if one minimizes only over Monge couplings, cf.\ \cite{Pr07b}.

Before defining the adapted Wasserstein distance between measures $\P$ and $\Q$ on $\Omega$, let us hint why distances related to weak convergence are not suitable for the results we have in mind. 
Assume for example that we are interested in a utility maximization problem in two periods and that Figure \ref{fig:usual.wasserstein} describes the laws $\P, \Q$ of two traded assets. 
Clearly they are very close in Wasserstein distance, as follows from considering the obvious Monge coupling induced by $T: \Omega \to \Omega, T(\P)= \Q$ depicted in Figure 1. At the same time, the outcome of utility maximization is certainly very different. Similarly, $\P$ is a martingale measure while $\Q$ allows for arbitrage. The clear reason for that is the different structure of information available at time $1$.
% \vspace{-15mm}
% \begin{figure}[H]
%    \centering
%    \includegraphics[page=1,width=0.45\textwidth]
%        {Adap.pdf} 
%  \caption{Close in Wasserstein, very different for utility maximization.}
% \label{fig:usual.wasserstein}
% \end{figure}
 \begin{figure}[H]
 \centering
\begin{tikzpicture}[scale=0.8]
\draw[thick,->] (0,0)--(2.2,0);
\draw[thick,->] (0,0)--(0,2.2);
\draw[thick,blue] (0,1.02) -- (1,1.02)--(2,2);
\draw[thick,red] (0,1.0) -- (1,1.0)--(2,0); 
\draw[fill] (1,1.01) circle [radius=0.05];
\draw[fill] (2,0) circle [radius=0.05];
\draw[fill] (2,2) circle [radius=0.05];
\draw[thick,->] (5,0)--(7.2,0);
\draw[thick,->] (5,0)--(5,2.2);
\draw[thick,blue] (5,1.02) -- (6,1.12)--(7,2);
\draw[thick,red] (5,1.0) -- (6,0.9)--(7,0); 
\draw[fill] (6,1.12) circle [radius=0.05];
\draw[fill] (6,0.9) circle [radius=0.05];
\draw[fill] (7,0) circle [radius=0.05];
\draw[fill] (7,2) circle [radius=0.05];
\draw[->] (2.5,1) -- (4.5,1);
\node [above] at (3.5,1) {T};
\end{tikzpicture}
 \caption{Map $T$ sends the blue path on the left, to the blue path on the right, and similarly for the red paths. The stochastic processes depicted are close in Wasserstein sense, but very different for utility maximization.}
 \label{fig:usual.wasserstein}
\end{figure}
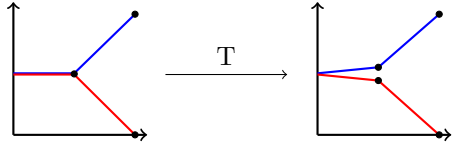
 \vspace{-3mm}
To exhibit why the Wasserstein distance does not reflect this different structure of information, let us review the transport condition $T(\P)= \Q$. We  rephrase it as 
\begin{align}\label{MassPreservation}
(T_1(X_1, X_2), T_2(X_1, X_2)) \quad  \sim  \quad (Y_1, Y_2).
\end{align}
While this condition  is of course perfectly natural in mass transport, \eqref{MassPreservation} almost seems like cheating when viewed from a probabilistic perspective: the map $T_1 $ should not be allowed to consider the future value $X_2$ in order to determine $Y_1$. To define an adapted version of the Wasserstein distance, the `process' $(T_i)_{i=1,2}$ should be taken to be \emph{adapted} in order to account for the different information structures of $\P$ and $\Q$. 

Naturally our official definition of adapted Wasserstein distances will not refer to adapted Monge transports but rather to  couplings which are `adapted' in an appropriate sense. Following Lassalle \cite{Las18}, we call such couplings (bi-)causal. Since the definition below may appear a bit technical at first glance, the following may be reassuring: 
In the discrete time setting and for absolutely  continuous measures $\P$, the weak closure of the set of adapted Monge couplings, i.e.\ $\pi= (\id, T)(\P)$ for $T$ adapted, is precisely the set of all causal couplings, see \cite{La18}.

\begin{definition}[(bi-)causal couplings] 
\label{def:adatped.coupling}
	For a coupling $\pi$ of $\mathbb{P}, \mathbb{Q}\in\mathcal{P}(\Omega)$ denote by $\pi(d\omega,d\eta)=\mathbb{P}(d\omega)\pi_\omega(d\eta)$ a regular disintegration w.r.t.~$\mathbb{P}$. 
	The set $\cpla(\P, \Q)$ of \emph{causal couplings} consists of all $\pi\in \cpl(\P, \Q)$ such that  for all $t\in I$ and $A\in\mathcal{F}_t$
	\[ \omega\mapsto \pi_\omega(A)  \quad\text{is }\mathcal{F}^{\P}_t\text{-measurable}.\]	
	The set of all \emph{bi-causal} couplings $\cplba(\P, \Q)$ consists of all $\pi\in\cpla(\P, \Q)$ such that also $S(\pi) \in \cpla(\Q, \P)$, where $S:\Omega\times \Omega \to \Omega \times \Omega, S(\omega,\eta):= (\eta,\omega)$. 
\end{definition}

In discrete time, a coupling $\pi$ is causal if and only if
\begin{align*} \pi\big( (Y_1,\dots,Y_t)\in A | X\big)& = \pi\big( (Y_1,\dots,Y_t)\in A | X_1,\dots X_t \big) ,
\end{align*} $\P$-a.s.\ 
for every $t$ and Borel set $A\subset \mathbb{R}^t$, that is, at time $t$,  given the past $(X_1, \ldots, X_t) $ of $X$, the distribution of $Y_t$ does not depend on the future $(X_{t+1}, \ldots, X_N)$ of $X$. 

\medskip 

Replacing  couplings by bi-causal couplings in \eqref{eq:UsualW} one arrives at the nested distance as introduced by Pflug and Pichler \cite{Pf09,PfPi12}. Since our goal is to compare also semimartingale models in continuous time we will work with an adapted Wasserstein distance  that is defined slightly differently. (Notably, {it is straightforward that} the two distances are equivalent for probabilities on $\R^N$. We will elaborate in Section \ref{sec:ChoiceOfCost} below, why the definition in \eqref{eq:AdW} is more appropriate for our purposes even in discrete time.)

In continuous time, we denote by $\mathcal{SM}(\Omega)$ the set of all probabilities $\mathbb{P}$ on (the Borel $\sigma$-field of) $\Omega$ under which the canonical process $X$ is a continuous semimartingale. In discrete time, $\mathcal{SM}(\Omega)$ denotes the set of all Borel probabilities $\mathbb{P}$ on $\Omega$ under which $X$ is integrable. In either case we can uniquely decompose $X=M+A$, with $A$ a finite variation predictable process started at zero, and $M$ a local martingale. Indeed, in the first case $X$ is a special semimartingale and in fact $M$ and $A$ are continuous too, and in the second case this is the Doob decomposition of an integrable adapted discrete-time process. For $p\in[1,\infty)$ we denote by $\mathcal{SM}_p(\Omega)$ the subset of $\mathcal{SM}(\Omega)$ for which \[\mathbb{E}_\mathbb{P}[\langle M\rangle_T^{p/2} +|A|_{\text{1-var}}^p]<\infty,\]where $\langle \cdot \rangle$ is the quadratic variation and $|\cdot|_{\text{1-var}}$ the first variation norm. Note also that by the BDG inequality $E_\mathbb{P}[\sup_{s\leq T} |M_s|] < \infty$ for $\mathcal{SM}_p(\Omega)$, hence $M$ is a true martingale.  

%Denote by $\mathcal{SM}(\Omega)$ the set of all probabilities $\mathbb{P}$ on (the Borel $\sigma$-field of) $\Omega$ under which the canonical process $X$ is a semimartingale, and  for $p\in[1,\infty)$ by $\mathcal{SM}_p(\Omega)$ the subset thereof for which  \[\mathbb{E}_\mathbb{P}[\langle M\rangle_T^{p/2} +|A|_{\text{1-var}}^p]<\infty.\]
%Here $X=M+A$ denotes the unique continuous semimartingale decomposition of $X$ under $\mathbb{P}$ into a continuous martingale $M$ starting in zero and a  continuous adapted process $A$ of finite variation, $\langle \cdot \rangle$ is the quadratic variation and $|\cdot|_{\text{1-var}}$ the first variation norm. Of course all integrable adapted processes are semimartingales in the discrete time case.  

\begin{definition}[Adapted Wasserstein distance]
\label{def:adapted.wasserstein.discrete}\label{AWD}
	For  $\mathbb{P},\mathbb{Q}\in\mathcal{SM}_p(\Omega)$, $p\geq 1$ set 
	\begin{align}\label{eq:AdW} 
	\mathcal{AW}_p(\mathbb{P},\mathbb{Q}):=\inf\Big\{ \mathbb{E}_\pi\big[ \langle M^X-M^Y\rangle_T^{p/2} + |A^X-A^Y|_{\text{1-var}}^p\big]^{1/p} : \pi \in \cplba(\P,\Q)\Big\},
	\end{align}
	where  $X=M^X+A^X, Y=M^Y+A^Y$ denote the semimartingale decomposition of $X$  and $Y$ resp.
\end{definition}

It is shown in Lemma \ref{lem:semimartingales.remain.semimartingale} that $\mathcal{AW}_p$ is well-defined (i.e.~that $X-Y$ is a semimartingale under every bi-causal coupling) and in Lemma \ref{lem:metric} that $\mathcal{AW}_p$ in fact defines a metric.
\begin{remark}
In the continuous time setup, the adapted Wasserstein distance can also be computed through 
\begin{align*}
\mathcal{AW}_p(\mathbb{P},\mathbb{Q})=\inf\Big\{ \mathbb{E}_\pi\big[ \langle X-Y\rangle_T^{p/2} + \mbox{MV}_T[|X-Y|^p\big]^{1/p} : \pi \in \cplba(\P,\Q)\Big\}.
\end{align*} Here $\mbox{MV}$ denotes the mean variation, i.e.\ $\mbox{MV}_T [Z]= \sup_\Delta \sum_{t_j\in\Delta}| \E[Z_{t_{j+1}} - Z_{t_j}|\F_{t_j}]|$, where the supremum is taken over all finite partitions $\Delta$ of $[0,T]$.  
\end{remark}

In Section \ref{subsec:RemarksAndExamples} below we will give explicit formulae for the adapted Wasserstein distance in the case of semi-martingale measures described by simple SDEs.

\subsection{Stability of Superhedging}

For the rest of this article, fix some $k\in\mathbb{R}_+$ and let $\mathcal{H}_k$ be the set of all predictable processes $$H\colon\Omega\times I\to[-k,k].$$ 
For every $p\geq 1$, write $b_p$ for the `upper' Burkholder-Davis-Gundy (BDG) constant, cf.\ Remark \ref{BDGremark} below. In particular it is known that $b_1 \leq 6 $ and that $b_2=2$.

Our first main result concerns the stability of superhedging and constitutes a stronger version of Theorem \ref{thm:hedging.with.loss.linear.intro} stated above.

\begin{theorem}
\label{thm:hedging.with.loss.linear}
	Let $\mathbb{P},\mathbb{Q}\in\mathcal{SM}_1(\Omega)$, $H \in \mathcal {H}_k$ and let $C\colon\Omega\to\mathbb{R}$ be Lipschitz with constant $L$.
	Then the hedging error under $\Q$ is bounded by the distance of $\P$ and $\Q$ plus the hedging error under $\P$ in the following sense:
	there exists $G\in\mathcal{H}_k$ such that %$m,G$ constitute an ``almost'' $\mathbb{Q}$-superhedge, i.e. 
	\begin{align}\label{eq almost superhedge}\tag{WHI}
	\begin{split}
	\mathbb{E}_\mathbb{Q}[(C-m-(G \bigcdot X)_T)^+]
	&\leq \mathbb{E}_\mathbb{P}[(C-m-(H\bigcdot X)_T)^+] \\
	&\quad+ b_1(k+L)\mathcal{AW}_1(\mathbb{P},\mathbb{Q}).
	\end{split}
	\end{align}
	Assume in addition that $H_t\colon\Omega\to\mathbb{R}$ is Lipschitz with constant $\tilde{L}$ for every $t\in I$.
	Then we can take $G=H$ and obtain 
	\begin{align}\label{eq almost superhedge strong}\tag{SHI}
	\begin{split}
	\mathbb{E}_\mathbb{Q}[(C-m-(H \bigcdot X)_T)^+]
	&\leq \mathbb{E}_\mathbb{P}[(C-m-(H\bigcdot X)_T)^+] \\
	&\quad + b_1(k+L)\mathcal{AW}_1(\mathbb{P},\mathbb{Q}) + { \beta}\mathcal{AW}_2(\mathbb{P},\mathbb{Q}),\end{split}
	\end{align}
	where $\beta:=2\sqrt{2}b_1 \tilde L \min\{ \mathcal{AW}_2(\mathbb{P},\delta_0),\mathcal{AW}_2(\mathbb{Q},\delta_0)\}$.
\end{theorem}
Importantly, it is impossible to transfer a superhedge under $\mathbb{P}$ into a superhedge under $\mathbb{Q}$.
This occurs already in a one-period framework and is not a by-product of our definition of adapted Wasserstein distance; see Remark \ref{rem:superhedging.does.not.work}.
A similar reasoning requires to consider only trading strategies bounded by $k$; see Remark \ref{rem:bounded.strategies}.

\medskip

It is worthwhile to compare the inequalities \eqref{eq almost superhedge} and \eqref{eq almost superhedge strong}:
\begin{itemize}
\item[(S)]
In a certain sense the  `strong hedging inequality' \eqref{eq almost superhedge strong}  seems to be the more relevant assertion: after all a  trader does not know that the model $\Q$ (rather than the model $\P$) describes reality and hence she might (somewhat stubbornly) stick to the initial plan of hedging her risk according to the strategy $H$. The inequality \eqref{eq almost superhedge strong} then 
allows to quantify the losses due to this model-error.
\item[(W)] However, the `weak hedging inequality' 
\eqref{eq almost superhedge} also has a particular merit: suppose that a trader $W$ starts with the prior belief that the asset price evolves according to a Black-Scholes model with volatility $\sigma_1$ but soon after time $0$ realizes that a volatility $\sigma_2$ (where $\sigma_2 \neq \sigma_1$) yields a more adequate description of reality. If the witty trader  $W$ makes an accurate guess about the correct model and updates her trading strategy accordingly, her losses can be controlled through the tighter bound in \eqref{eq almost superhedge}.
\end{itemize}

In Theorem \ref{prop:hedging.loss} we provide a version of Theorem \ref{thm:hedging.with.loss.linear}, where $(\cdot)^+$ is replaced by a convex, strictly increasing loss function $l\colon\mathbb{R}\to\mathbb{R}_+$. 

Another way to gauge the effectiveness of an almost superhedge is by means of risk measures. We postpone the general formulation to Theorem \ref{prop:heding.law.invariant} and 
first present a version that appeals to the average value of risk $\mathrm{AVaR}^\mathbb{P}_\alpha$.
Recall that
for a random variable $Z\colon\Omega\to\mathbb{R}$ 
\[\mathrm{AVaR}^\mathbb{P}_\alpha(Z):=\inf_{m\in\mathbb{R}} \mathbb{E}_\mathbb{P}[(Z-m)^+/\alpha +m],\]
is the average value at risk at level $\alpha\in(0,1)$ under model $\mathbb P$. We then have
\begin{theorem}
\label{thm:avar}
	Assume that $C\colon\Omega\to\mathbb{R}$ is Lipschitz with constant $L$.
	Then 
	\begin{align*}
	&\Big|\inf_{H\in\mathcal{H}_k} \mathrm{AVaR}^\mathbb{P}_\alpha(C-(H\bigcdot X)_T)-\inf_{H\in\mathcal{H}_k} \mathrm{AVaR}^\mathbb{Q}_\alpha(C-(H\bigcdot X)_T)\Big| 
	\leq r\, \mathcal{AW}_{1}(\mathbb{P},\mathbb{Q}),
	\end{align*}
	for $r:=b_1(L+k)/\alpha $. If $H\in\mathcal{H}_k$ is such that $H_t\colon\Omega\to[-k,k]$ is Lipschitz with constant {$\tilde L$} for every $t\in I$ and $\beta$ is the constant defined in Theorem \ref{thm:hedging.with.loss.linear}, then
	\begin{align*}
	\big|\mathrm{AVaR}^\mathbb{P}_\alpha(C-(H\bigcdot X)_T)-\mathrm{AVaR}^\mathbb{Q}_\alpha(C-(H\bigcdot X)_T)\big|
	\leq r\, \mathcal{AW}_{1}(\mathbb{P},\mathbb{Q})+\frac\beta\alpha \mathcal{AW}_2(\mathbb{P},\mathbb{Q}).
	\end{align*}
\end{theorem}

The interpretation of this result is similar to the one of Theorem \ref{thm:hedging.with.loss.linear}:
As $\mathrm{AVaR}^\mathbb{P}_\alpha(\cdot)$ is translation invariant, one has
\begin{align*}
\inf_{H\in\mathcal{H}_k} \mathrm{AVaR}^\mathbb{P}_\alpha(C-(H\bigcdot X)_T)
=\inf\Big\{m\in\mathbb{R} :\begin{array}{l} 
\text{there is } H\in\mathcal{H}_k \text{ such that}\\
\mathrm{AVaR}_\alpha^\mathbb{P}(C-m-(H\bigcdot X)_T)\leq 0 
\end{array}\Big\},
\end{align*}
and the right-hand side constitutes a relaxed version of the superhedging price.
 
\medskip 

 Notably, the explicit calculations of adapted Wasserstein distance given in  Section \ref{subsec:RemarksAndExamples} imply that Theorem \ref{thm:avar} (and similarly Theorem \ref{thm:hedging.with.loss.linear}) are sharp 

\begin{example}[Hedging in a Brownian framework]
\label{ex:hedge.BM}
	Consider a European call option $C(X)=(X_T-K)^+$, where for simplicity $K=0$.
	Moreover, let $\mathbb{P}^\sigma$ be Wiener measure with constant volatility $\sigma\geq0$.
	Then for every $\sigma,\hat{\sigma}\geq0$, $k\geq 1$, and $\alpha\in(0,1)$ it holds that (we defer the proof of this fact to Section \ref{sec proofs intro})
	\begin{align*}
	&\Big|\inf_{H\in\mathcal{H}_k} \mathrm{AVaR}_\alpha^{ \mathbb{P}^{\sigma} }(C-(H\bigcdot X)_T)-\inf_{H\in\mathcal{H}_k} \mathrm{AVaR}_\alpha^{ \mathbb{P}^{\hat{\sigma}} }(C-(H\bigcdot X)_T)\Big|\\
	=&\big|\mathbb{E}_{\mathbb{P}^\sigma}[C] - \mathbb{E}_{ \mathbb{P}^{\hat{\sigma}} }[C] \big| 
	=\frac{1}{\sqrt{2\pi}} T|\sigma-\hat{\sigma}|
	=\frac{1}{\sqrt{2\pi}} \mathcal{AW}_1(\mathbb{P}^{\sigma},\mathbb{P}^{\hat{\sigma}}).
	\end{align*}
	This shows that the estimate in Theorem \ref{thm:avar} is tight (up to constants), in the sense that it is essentially impossible to improve on the probability metric $\mathcal{AW}_1$.
\end{example}

We make the important remark that  Glanzer, Pflug, and Pichler \cite{GlPfPi17} use the nested distance to control acceptability prices in discrete time models in a Lipschitz fashion through the nested distance of these models. Specifically, in a discrete time one-period framework  \cite[Proposition 3]{GlPfPi17}
 and   Theorem \ref{thm:avar} yield almost the same assertion: in this setup, the only difference is that \cite[Proposition 3]{GlPfPi17} does not specify a Lipschitz constant and does not assume uniform boundedness of the admissible hedging strategy. (However, the latter seems to be in conflict with our Remark \ref{rem:bounded.strategies} below.)

\subsection{Stability of Utility Maximization and Utility Indifference Pricing}
We move on to consider the continuity of utility maximization.
Let $U\colon\mathbb{R}\to\mathbb{R},$ be a utility function which is concave, increasing, and denote by $U'$ the left-continuous version of the derivative.
We have

\begin{theorem}
\label{thm:utility}
	Let $C\colon \Omega\to\mathbb{R}$ be Lipschitz continuous and assume that there exists $c\geq 0$ such that $U'(x)\leq c(1+|x|^{p-1})$ for all $x$.
	Then, for every $R\geq 0$ there exists a constant $K$ %depending on $k,p,U,C,R$ 
	such that 
	\[\Big| \sup_{H\in\mathcal{H}_k} \mathbb{E}_\mathbb{P}[U(C+(H\bigcdot X)_T)] - \sup_{H\in\mathcal{H}_k} \mathbb{E}_\mathbb{Q}[U(C+(H\bigcdot X)_T)]\Big|
	\leq K \cdot \mathcal{AW}_p(\mathbb{P},\mathbb{Q}),\]
	for all $\mathbb{P},\mathbb{Q}\in\mathcal{SM}_p(\Omega)$ with $\mathcal{AW}_p(\mathbb{P},\delta_0),\mathcal{AW}_p(\mathbb{Q},\delta_0)\leq R$.
\end{theorem}

{The failure of usual Wasserstein distances to guarantee stability of utility maximization is illustrated in Remark \ref{ex:wasserstein.does.not.work}.}

%\comment{dani:somebody please check that this is really utility indifference pricing. JB: I did, looks fine.}
%Complementary to the problem of risk-based superhedging, another popular way

A common way of quantifying the value of a claim is via utility indifference pricing:\footnote{{We are grateful to the anonymous referee for pointing out that we could include the stability of utility indifference pricing w.r.t.\ adapted Wasserstein distance.}} given a claim $C$, the utility indifference (bid-) price $v$ is defined as the solution of the following equation
\[ \sup_{H\in\mathcal{H}_k} \mathbb{E}_\mathbb{P}[U(C-v+(H\bigcdot X)_T)]
=\sup_{H\in\mathcal{H}_k} \mathbb{E}_\mathbb{P}[U( (H\bigcdot X)_T)].\]
Continuing in the spirit of the present paper, we are interested in the stability of $\mathbb{P}\mapsto v(\mathbb{P})$, where the latter denotes the utility indifference price associated to the model $\mathbb{P}$.
\begin{theorem}
\label{thm:utility.indifference}
	Let $C\colon \Omega\to\mathbb{R}$ be Lipschitz continuous and  assume that there exists $c\geq 0$ such that $0<U'(x)\leq c(1+|x|^{p-1})$ for all $x$.
	Then, for every $R\geq 0$ there exists a constant $K$ %depending on $k,p,U,C,R$ 
	such that 
	\[\Big| v(\mathbb{P})-v(\mathbb{Q})\Big|
	\leq K \cdot \mathcal{AW}_p(\mathbb{P},\mathbb{Q}),\]
	for all $\mathbb{P},\mathbb{Q}\in\mathcal{SM}_p(\Omega)$ with $\mathcal{AW}_p(\mathbb{P},\delta_0),\mathcal{AW}_p(\mathbb{Q},\delta_0)\leq R$.
\end{theorem}
%\marginpar{MB: Maybe Dani can add a remark stating that the Lipschitz constant in the two Thms depends only on $c, p, U'(0), R, $ etc.}

\subsection{Structure of the paper}
In Section \ref{sec:literature} we briefly review the literature related to this paper.  In Section \ref{sec:AWSection} we establish some basic properties of the adapted Wasserstein distance, discuss the choice of cost function and give some examples. Moreover we derive a  contraction principle (Theorem \ref{thm:contraction}) which relates adapted Wasserstein distance with a `weak' (in the sense of Gozlan et al \cite{GoRoSaTe17}) transport distance. This result forms the basis for the  proofs of the results mentioned in the introduction, as well as certain extensions of these results, see Section \ref{sec proofs intro}. Finally we conclude with some remarks in Section \ref{sec the end}.

\section{Literature}\label{sec:literature}

{
The articles closest in spirit to ours are \cite{AcBaZa16,BiTa19,GlPfPi17}. Acciaio, Zalashko and one of the present authors consider in \cite{AcBaZa16} an object related to the adapted Wasserstein distance in continuous time in connection with utility maximization, enlargement of filtrations and optimal stopping. Glanzer, Pflug, and Pichler \cite{GlPfPi17} prove a deviation-inequality for the so-called nested distance in a discrete time framework\footnote{Note added in revision: improved convergence rates have been recently obtained in \cite{BaBaBeWi20} for a related sample-based estimator. Together with the results of the present article, this gives statistical consistency for an empirical version of the financial problems considered.}, and consider acceptability pricing  over an ambiguity set described through the nested distance. Bion-Nadal and Talay  \cite{BiTa19} study via PDE arguments a continuous-time optimization problem which is related to the adapted Wasserstein distance.

The concept of causal couplings, and optimal transport over causal couplings, has been recently popularized by Lassalle \cite{Las18} although precursors can be found in the works \cite{YW,Rueschendorf}. This notion is central to the recent articles \cite{AcBaZa16,BaBeLiZa16,BaBeEdPi17,BaBeHuKa17}.

The idea of strengthening weak convergence of measures in order to account for the temporal evolution has some history. Indeed several authors have independently introduced different approaches to address this challenge: 
The seminal unpublished work by Aldous \cite{Al81} introduces the notion of extended weak convergence for the study of stability of optimal stopping problems. The principal idea is not to compare the laws of processes directly, but rather the laws of the corresponding prediction processes.
 Independently,  Hellwig \cite{He96} introduces the information topology for the stability of equilibrium problems in economics. Roughly, two probability measures on a product of finitely many spaces $X_1\times \ldots \times X_N$ are considered to be close if for each $t \leq N$ the projections onto the first $t$ coordinates as well as the corresponding conditional (regular) disintegrations are close. Unrelated to these developments Pflug and Pichler \cite{Pf09,PfPi12,PfPi14} have introduced the nested distances for the stability of stochastic programming in discrete time. The nested distance is the obvious role model for the adapted Wasserstein distances considered in this article and (as mentioned above) for a fixed number of time steps and $p \geq 1$, they are obviously equivalent. Yet another idea to account for the temporal evolution of processes would be to symmetrise the causal transport costs $\mathcal{W}_{c}(\P,\Q)$ defined by  Lassalle \cite{Las18} by taking the maximum or sum of $\mathcal{W}^2_{c}(\P,\Q)$ and $\mathcal{W}^2_{c}(\Q,\P)$; this was pointed out by Soumik Pal.
 
 In parallel work \cite{BaBaBeEd19b}, the four authors of the present article investigate the relations between these concepts in detail. Remarkably, in discrete time \emph{all} of the concepts mentioned above (adapted Wasserstein distances, extended weak convergence, information topology, nested distances, symmetrised causal transport costs) define the \emph{same topology}. {As noted above, this `weak adapted topology' refines the usual weak topology (properly for $T\geq 2$, see also Remark \ref{rem:superhedging.does.not.work}). The articles \cite{BaBeEdPi17,BaBaBeEd19b,Ed19} investigate basic properties of this topology, e.g.\ the weak adapted topology is Polish \cite[Section 5]{BaBeEdPi17}, sets are totally bounded w.r.t.\ to adapted Wasserstein distance / nested distance if and only if they are totally bounded w.r.t.\ usual Wasserstein distance \cite[Lemma 1.6]{BaBaBeEd19b}. 
  For recent applications of these concepts to optimal transport and probabilistic variants thereof we refer to \cite{BaBePa18,BaPa19,Wi19}.}

  {In contrast, fundamental topological properties of the above mentioned concepts in the  continuous time case seem to be much less understood and, as far as the authors are concerned, pose an interesting challenge for future research. Specifically, it is not clear to us whether the topology associated to the adapted Wasserstein distance is  Polish  in the continuous time case. In a similar vein, we expect that  results analogous to the ones of the present article should apply in the case of c\`adl\`ag paths, but this extension is beyond the scope of our current understanding of adapted Wasserstein distances.}

The question of stability in mathematical finance has been studied from different perspectives over the years.
Notably, starting with the articles of Lyons  \cite{Ly95} and  Avellaneda, Levy, Paras \cite{AvLePa95} the area of robust finance has mainly focused on extremal models and hedging strategies which dominate the payoff for every model in a specified class. Following the publication of Hobson's seminal article \cite{Ho98a} connections with the Skorokhod embedding problem have been a driving force of the field, see the surveys of Hobson \cite{Ho11} and \OB\  \cite{Ob04}. Recently this has been complemented by techniques coming from (martingale) optimal transport, early papers which advance this viewpoint include \cite{HoNe12,BeHePe12,GaHeTo13,BeJu16,BoNu13,DoSo12,CaLaMa14,BeCoHu16}.
The literature on `local' misspecification of volatility in a sense more closely related to the present article appears more spare.
  El Karoui, Jeanblanc, and Shreve \cite{ElJeSh98}  establish in a stochastic volatility framework that if the misspecified volatility dominates the true volatility, then the misspecified price of call options dominates the real price; see also the elegant account of 
 Hobson \cite{Ho98c}. 
 More recently, the question of pricing and hedging under uncertainty about the volatility of a reference local volatility model  is studied by  Herrmann, Muhle-Karbe, and Seifried  \cite{HeMuSe17} (see also \cite{HeMu17}). Less plausible models are penalized through a mean square distance to the volatility of the reference model and the authors obtain explicit formulas for prices and hedging strategies in a limit for small uncertainty aversion. 
Becherer and Kentia \cite{BeKe17} derive worst-case  good-deal bounds under model ambiguity which concerns drift as well as volatility. Indeed, discussions with Dirk Becherer motivated us to consider also models with drift in our results on stability of super hedging.
 The behaviour of the superhedging price in a ball (w.r.t.\ various notions of distance) around a reference model is studied in depth by \OB{} and Wiesel \cite{ObWi18} for a $d$-dimensional asset and one time period.

A notable implication of our work is that it yields a coherent way to measure model-uncertainty (in the sense of Cont's influential article \cite{Co06}):  Fix  a subset $M_0$ of the set $M$ of all consistent models, i.e.\ martingale measures which are consistent with  benchmark instruments whose price can be observed on the market.
 Given $M_0$, the model uncertainty associated to a derivative $f$ can be gauged through $$\rho_{M_0}(f):=\sup \{ \mathbb E_{\mathbb Q} f :\mathbb Q\in M_0 \} - \inf \{\mathbb E_{\mathbb Q}
f : \mathbb Q\in M_0 \}.$$ The worst-case approach typically pursued in robust finance then yields  $\rho_{M_0}(f)$ for $M_0=M$, but it appears equally natural to take $M_0$ to be an infinitesimal ball around a reference model. This approach is first carried out by Drapeau, Ob{\l}{\'o}j, Wiesel and one of the present authors \cite{BaDrObWi19} in a one period framework. Our results  indicate that adapted Wasserstein distance provides a way to extend this to a multi-period setup, and we intend to pursue this further in future work.
 
On a different note, much work has been done regarding the convergence of discrete time models to their continuous time analogues. Due to the vastness of this literature we refer the reader to the book \cite{prigent2003weak} for references. Finally, in more recent times and starting from the works of Kardaras and {\v{Z}}itkovi{\'c}, the stability of utility maximization has been studied in \cite{ZitKar,Larsenpref,ZitLar,MWstability,weston2016stability} among others.}

\section{The adapted Wasserstein distance}\label{sec:AWSection}

\subsection{Basic properties of $\mathcal{AW}_p$}
The following Lemma shows that $\mathcal{AW}_p$ is well-defined.

\begin{lemma}
\label{lem:semimartingales.remain.semimartingale}
	Let $\mathbb{P},\mathbb{Q}$ be integrable (semi-)martingale measures for $X,Y\colon \Omega\to\Omega$, respectively, and let $\pi$ be a bi-causal coupling between $\mathbb{P}$ and $\mathbb{Q}$.
	Then $X,Y,X-Y\colon\Omega\times\Omega\to\Omega$ are (semi)-martingales w.r.t.~$\pi$.
	Further, if $X=M+A$ denotes the semimartingale decomposition under $\mathbb{P}$, then up to evanescence $M+A$ is the semimartingale decomposition of $X$ under $\pi$.  
\end{lemma}
\begin{proof}
	Let $X=M+A$ be the semimartingale decomposition under $\mathbb{P}$ and consider $M$ and $A$ as processes on $\Omega\times\Omega$ via $M(\omega,\eta):=M(\omega)$ and $A(\omega,\eta):=A(\omega)$.
	Further let $\pi=\mathbb{P}(d\omega)\pi_\omega(d\eta)$ be a bi-causal coupling between $\mathbb{P}$ and $\mathbb{Q}$.
	To show that $X=M+A$ remains the semimartingale decomposition under $\pi$, it is enough to show that $M$ is a martingale under $\pi$.
	To that end, let $0\leq s\leq t$ and let $Z\colon\Omega\times\Omega\to\mathbb{R}$ be $\mathcal{F}_s\otimes\mathcal{F}_s$-measurable and bounded. {(Recall that $\mathbb F=(\F_t)_t$ denotes the right-continuous filtration generated by $X$ and that we endow $\Omega \times \Omega$ with the filtration $(\F_t \otimes \F_t)_t$.)}
	Then the random variable $Z'\colon\Omega\to\mathbb{R}$ defined by 
	\[Z'(\omega):=\int Z(\omega,\eta)\,\pi_\omega(d\eta) \quad\text{is $\mathcal{F}^\P_s$-measurable,}\]
	 and clearly bounded.
	Indeed, if $Z(\omega,\eta)=Z^1(\omega)Z^2(\eta)$ for $\mathcal{F}_s$-measurable bounded functions $Z^1$ and $Z^2$, then it follows from the definition of bi-causality that $Z'$ is $\mathcal{F}^\P_s$-measurable; the general statement then follows from a monotone class argument.
	Therefore
	\begin{align*}
	\mathbb{E}_\pi[(M_t-M_s)Z]
	&=\int (M_t(\omega)-M_s(\omega))\int Z(\omega,\eta)\,\pi_\omega(d\eta)\,\mathbb{P}(d\omega)\\
	&=\mathbb{E}_\mathbb{P}[(M_t-M_s)Z']\\
	&=0,
	\end{align*}
	by the martingale property of $M$ under $\mathbb{P}$.
	This shows that $M$ is a martingale under $\pi$ and therefore that $X=M+A$ is the semimartingale decomposition under $\pi$.
\end{proof}

\begin{lemma}
\label{lem:metric}
	$\mathcal{AW}_p$ defines a metric on the set $\mathcal{SM}_p(\Omega)$.
\end{lemma}
{We note that very similar arguments could be used to show that $\mathcal{AW}_p$ defines a metric for semimartingales with infinite time horizon $\mathbb N$ or $[0,\infty)$.}
\begin{proof}[Proof of Lemma \ref{lem:metric}]
	It is clear that $\mathcal{AW}_p(\mathbb{P},\mathbb{Q})=\mathcal{AW}_p(\mathbb{Q},\mathbb{P})\geq0$ for all $\mathbb{P},\mathbb{Q}\in\mathcal{SM}_p(\Omega)$.
	Suppose that $\mathcal{AW}_p(\mathbb{P},\mathbb{Q})=0$. 
	As $\|\cdot\|_\infty\leq |\cdot|_{\text{1-var}}$, it is immediate that if $\pi$ participates in the infimum defining $\mathcal{AW}_p(\mathbb{P},\mathbb{Q})$, and $X-Y=M+A$, then
	\begin{align*}
	\mathbb{E}_{\pi}[ \|X-Y\|_\infty^p ]
	&\leq  2^{p-1} \mathbb{E}_\pi[ \| M\|_\infty^p + |A|_{\text{1-var}}^p]\\
	&\leq 2^{p-1} b_p  \mathbb{E}_\pi[ \langle M\rangle_T^{p/2} + |A|_{\text{1-var}}^p]
	\end{align*}
	where $b_p$ denotes the BDG constant  and we used the BDG inequality for the martingale $M$.
	Hence the usual Wasserstein distance between $\mathbb{P}$ and $\mathbb{Q}$ (defined w.r.t.~the $\|\cdot\|_\infty$-norm) is dominated from above by $\mathcal{AW}_p(\mathbb{P},\mathbb{Q})$, and so $\mathbb{P}=\mathbb{Q}$. 
	 		
	We now prove the triangle inequality. 
	Let $\mathbb{P},\mathbb{Q},\mathbb{R}$ given.
	We fix $\varepsilon>0$ and assume $\pi$ is bi-causal $\varepsilon$-optimal for $\mathcal{AW}_p(\mathbb{P},\mathbb{Q})$ and $\tilde{\pi}$ is bi-causal $\varepsilon$-optimal for $\mathcal{AW}_p(\mathbb{Q},\mathbb{R})$.
	{In the next couple of lines, $\omega$ will always denote the first coordinate of a vector in $\Omega^3$, $\eta$ the second, and $\gamma$ the last.}
	Let
	\[\pi(d\omega,d\eta)=\pi_\eta(d\omega)\,\mathbb{Q}(d\eta) \quad\text{and}\quad \tilde{\pi}(d\eta,d\gamma)=\tilde{\pi}_\eta(d\gamma)\mathbb{Q}(d\eta)\]
	be disintegrations, and define $\Pi\in\mathcal{P}(\Omega^3)$ by
	\[	\Pi(d\omega,d\eta,d\gamma)= \pi_\eta(d\omega)\,\tilde{\pi}_\eta(d\gamma)\,\mathbb{Q}(d\eta).\]
	If $\overline{\pi}(d\omega,d\gamma):=\int_\Omega \Pi(d\omega,d\eta,d\gamma)$ is the projection of $\Pi$ onto the first and third components, then it is clear that the first and second marginals of $\overline{\pi}$ are $\mathbb{P}$ and $\mathbb{R}$ respectively.
	Moreover, a disintegration of $\overline{\pi}=\overline{\pi}_\omega(d\gamma)\,\mathbb{P}(d\omega)$ is given by
	\[ \overline{\pi}_\omega(d\gamma)= \int_\Omega \tilde{\pi}_\eta(d\gamma)\,\pi_\omega(d\eta),\]
	{where, as indicated above, $\pi_\omega$ now denotes the disintegration of $\pi$ w.r.t.\ the first coordinate, that is  $\pi(d\omega,d\eta)=\pi_\omega(d\eta)\,\mathbb{P}(d\omega)$.
	We claim that, for every $A\in\mathcal{F}_t$, the mapping $\omega\mapsto \overline{\pi}_\omega(A)$ is $\mathcal{F}^\P_t$-measurable.
	Indeed, by bi-causality of $\tilde{\pi}$ one has that $\eta\mapsto\tilde{\pi}_\eta(A)$ is $\mathcal{F}^\Q_t$-measurable.
	Thus there is an $\mathcal{F}_t$-measurable function $X$ and a $\mathbb{Q}$-almost surely zero function $N$ such that $\tilde{\pi}_\eta(A)=X(\eta)+N(\eta)$ for all $\eta\in\Omega$.
	Then $\overline{\pi}_\omega(A)=\int_\Omega X(\eta)\,\pi_\omega(d\eta) + \int_\Omega N(\eta)\pi_\omega(d\eta)$ for all $\eta\in\Omega$.
	The first term is $\mathcal{F}^\P_t$-measurable (by bi-causality of $\pi$), and, as $\pi$ is a coupling between $\mathbb{P}$ and $\mathbb{Q}$, one has that $\int_\Omega N(\eta)\pi_\omega(d\eta)=0$ for $\mathbb{P}$-almost all $\omega\in\Omega$.}

	The argument for $\overline{\pi}=\overline{\pi}_\gamma(d\omega)\,\mathbb{R}(d\gamma)$ is similar and therefore $\overline{\pi}$ is a bi-causal coupling between $\mathbb{P}$ and $\mathbb{R}$. 
	Finally, it follows as in the proof of Lemma \ref{lem:semimartingales.remain.semimartingale} that, if $X=M^X+A^X$, $Y=M^Y+A^Y$, and $Z=M^Z+A^Z$ are the semimartingale decompositions under $\mathbb{P}$, $\mathbb{Q}$, and $\mathbb{R}$, then they remain the semimartingale decomposition under $\Pi$ on $\Omega^3$ endowed with the product filtration.

	To finish the proof of the triangle inequality, we observe that
	\begin{align*}
	\mathcal{AW}_p(\mathbb{P},\mathbb{R})
	&\leq \mathbb{E}_{\overline{\pi}}[ \langle M^X-M^Z \rangle_T^{p/2} + |A^X-A^Z|_{\text{1-var}}^p  ]^{1/p}\\
	&= \mathbb{E}_{\Pi}\Big[ \langle (M^X-M^Y)+(M^Y-M^Z) \rangle_T^{p/2} + \cdots\\
	&\quad \cdots + |(A^X-A^Y)+(A^Y-A^Z)|_{\text{1-var}}^p  \Big]^{1/p}.
	\end{align*}
	The function $M\mapsto \mathbb{E}_{\Pi}[ \langle M \rangle_T^{p/2}]^{1/p}$ is known to be a norm on the space $\mathcal M_p(\Pi)$ of $\Pi$-martingales started at zero whose supremum is $p$-integrable. 
	Likewise $A\mapsto \mathbb{E}_{\Pi}[ |A|_{\text{1-var}}^p ]^{1/p}$ is a norm on the space of finite variation processes with $p$-integrable variation. 
	Hence $$(M,A)\mapsto \|(M,A)\|:= \mathbb{E}_{\Pi}[\langle M \rangle_T^{p/2}+ |A|_{\text{1-var}}^p ]^{1/p}$$ is a norm on the product of these spaces. 
	We conclude the proof for the triangle inequality with
	\begin{align*}
	\mathcal{AW}_p(\mathbb{P},\mathbb{R})&
	\leq \|(M^X-M^Y,A^X-A^Y)+(M^Y-M^Z,A^Y-A^Z)\|
	\\ &\leq \|(M^X-M^Y,A^X-A^Y)\| + \|(M^X-M^Y,A^X-A^Y)\|
	\\&= \mathbb{E}_{\pi}[  \langle M^X-M^Y \rangle_T^{p/2}+ |A^X-A^Y|_{\text{1-var}}^p ]^{1/p} \\
		&\quad+ \mathbb{E}_{\tilde{\pi}} [  \langle M^Y-M^Z \rangle_T^{p/2}+ |A^Y-A^Z|_{\text{1-var}}^p ]^{1/p}
	\\&\leq 2\varepsilon+ \mathcal{AW}_p(\mathbb{P},\mathbb{Q}) + \mathcal{AW}_p(\mathbb{Q},\mathbb{R}),
	\end{align*}
	since the semimartingale decomposition of $X-Y$ under $\pi$ is $(M^X-M^Y)-(A^X-A^Y)$, with an analogous expression for $Y-Z$ under $\tilde{\pi}$.

	To conclude the proof, it remains to show that $\mathcal{AW}_p(\mathbb{P},\mathbb{Q})<\infty$ for all $\mathbb{P},\mathbb{Q}\in\mathcal{SM}_p(\Omega)$.
	By Lemma \ref{lem:semimartingales.remain.semimartingale}, we have $\mathcal{AW}_p(\mathbb{P},\delta_0)= \mathbb{E}_\mathbb{P}[\langle M\rangle_T^{p/2} + |A|_{\text{1-var}}^p]^{1/p}$ where $X=M+A$ is the semimartingale decomposition under $\mathbb{P}$.
	Therefore the triangle inequality implies that $\mathcal{AW}_p$ is real-valued on $\mathcal{SM}_p(\Omega)$.
\end{proof}

\subsection{Examples and explicit calculations}\label{subsec:RemarksAndExamples}

We start by a simple result which permits to give a closed-form expression of the adapted Wasserstein distance in given continuous-time situations: 

\begin{proposition}\label{lem synchron}
For $i\in\{1,2\}$ consider the SDEs with bounded progressive coefficients:
\begin{align}\label{eq two Brownians}
dX^i_t=\mu_i(t,\{X^i_s\}_{s\leq t})dt+\sigma_i(t,\{X^i_s\}_{s\leq t})dB^i_t.
\end{align}
Assume that each SDE admits a unique strong solution and denote by $\mathbb{P}^{\mu_i,\sigma_i}$ the respective laws. Further assume that
\begin{itemize}
\item $\mu_1$ is a function of time only (namely $\mu_1\colon[0,T]\to \mathbb R$)
\item $\sigma_1,\sigma_2\geq 0$ and at least one of them is a function of time only.
\end{itemize}
 Then the \textit{synchronous} coupling (namely $\pi^*=$ joint law of $(X^1,X^2)$, where $B^1=B^2$ in \eqref{eq two Brownians}), is optimal in the definition of $\mathcal{AW}_p(\mathbb{P}^{\mu_1,\sigma_1},\mathbb{P}^{\mu_2,\sigma_2})$. 
\end{proposition}

The discrete time version of the aforementioned synchronous coupling is given by the Knothe-Rosenblatt rearrangement \cite{BaBeLiZa16}, and a variant of the previous result can also be obtained in the discrete time framework.

\begin{proof}
Let $\pi$ be a feasible coupling for $\mathcal{AW}_p(\mathbb{P}^{\mu_1,\sigma_1},\mathbb{P}^{\mu_2,\sigma_2})$, leading to a finite cost. 
Naturally for this proof we denote the coordinate process on $\Omega\times \Omega$ by $(X^1,X^2)$. 
As before we let $X^i=A^i+M^i$ be the unique continuous semimartingale decomposition of $X^i$ under the $\mathbb{P}^{\mu_i,\sigma_i}$-completion of its right-continuous filtration. 
Observe that $\frac{d}{dt}A^1$ is a.s.\ deterministic, by the assumption on $\mu_1$, and that the law of $\frac{d}{dt}A^2$ is independent of the coupling $\pi$. 
Both facts can be derived easily from the identity
$$\frac{d}{dt}A^i_t= \lim_{\epsilon\searrow 0}  \frac{\mathbb{E}_\pi\big[ X^i_{t+\epsilon} | \mathcal F^{X^i}_t \big] -X^i_t}{\epsilon},$$
which by Lebesgue differentiation theorem holds $dt\otimes d\pi$-a.s. 
As a consequence, the term $\mathbb{E}_\pi[|A^1-A^2|^p_{1-var}]$ is independent of the coupling $\pi$ and so we may ignore it and only focus on the term $\mathbb{E}_\pi[\langle M^1-M^2\rangle_T^{p/2}]$.

By Doob's martingale representation \cite[Theorem 4.2]{KaSh12}, in a possibly enlarged filtered probability space $(\tilde \Omega,\tilde{\mathcal F},\tilde \pi)$ we may represent the martingale $(M^1,M^2)$ by
\begin{align*}
M^i_t=\int_0^t  \sigma_{i1}dW + \int_0^t  \sigma_{i2}d\hat W,
\end{align*}
where $W,\hat W$ are independent standard one-dimensional Brownian motions and $\{\sigma_{ik}:i,k\in\{1,2\}\}$ real-valued processes, both of them adapted in the enlarged filtered space. 
In the following we will omit the argument $\{X^i_s\}_{s\leq t}$ from $\sigma_i$. 
Necessarily $$\sigma_i^2=\frac{d}{dt}\langle M^i\rangle_t=\sigma_{i1}^2+\sigma_{i2}^2\,,\,\, (dt\otimes d\tilde \pi-a.s.).$$
By Cauchy-Schwarz inequality we deduce that almost surely
$$\langle M^1,M^2\rangle_T = \int_0^T[\sigma_{11}\sigma_{21}+\sigma_{12}\sigma_{22}]dt \leq \int_0^T \sigma_1\sigma_2 dt, $$ and accordingly we get the lower bound
$$\mathbb{E}_{\pi}[\langle M^1-M^2\rangle_T^{p/2}]
\geq \mathbb{E}_{\pi}\Big[\Big( \int_0^T (\sigma_1 -\sigma_2)^2 dt \Big)^{p/2}\Big].$$
As in the beginning of the proof, the right-hand side does not depend on the coupling $\pi$ thanks to either $\sigma^i$ being a function of time only. To conclude observe that for the synchronous coupling $\pi^*$ we have equality in the above equation.
\end{proof}

As an easy consequence we have

\begin{example}
\label{ex det diff}
	For bounded Lipschitz functions $\mu_1,\mu_2,\sigma_1,\sigma_2$ we denote by $\mathbb{P}^{\mu_i,\sigma_i}$ the law of the diffusion
	\[dX^i_t=\mu_i(t,X^i_t)dt+\sigma_i(t,X^i_t)dB_t.\]
	Assume that
	\begin{itemize}
	\item $\mu_{i}$ is independent of the $x$-variable, some $i\in\{1,2\}$, and
	\item $\sigma_k$ is independent of the $x$-variable, some $k\in\{1,2\}$.
	\end{itemize}
	Calling $j\in\{1,2\}\backslash\{i\}$ and $\ell\in\{1,2\}\backslash\{k\}$, we have
	\begin{align*}
	\mathcal{AW}_p(P^{\mu_1,\sigma_1},P^{\mu_2,\sigma_2})^p 
	&= \mathbb{E}\Big[ \Big( \int_0^T [\sigma_\ell(t,X_t^\ell)-\sigma_k(t)]^2 dt\Big)^{p/2} \Big]\\
	&\quad+ \mathbb{E}\Big[ \Big(\int_0^T |\mu_j(t,X^j_t)-\mu_i(t)|dt   \Big)^{p} \Big].
	\end{align*}
\end{example}

We now illustrate that in general it is not true that the straightforward synchronous coupling of Proposition \ref{lem synchron} is optimal. As a consequence, we do not expect a closed-form expression for the adapted Wasserstein distance. A discrete-time version of this observation is discussed in \cite[Section 7]{BaBeEdPi17}.

\begin{example}
	Consider $d=1$, $T=2$, and for each $c\in\mathbb R$ introduce
	\[\mu^c_t(\omega):=c 1_{[1,2]}(t) \mathop{\mathrm{sign}}(\omega_1)
	\quad\text{and}\quad
	\hat{\mu}^c_t(\omega):=-\mu^c_t(\omega).\]
	Assuming that $B$ is a Brownian motion, and for $\sigma\in\mathbb R_+$, we introduce the couplings
	\begin{align*}
	\pi_1&:= \mathop{\mathrm{Law}} \Big(\sigma B + \int \mu^c_t(B)dt\, ,\, \sigma B + \int \hat{\mu}^c_t(B)dt \Big),\\ 
	\pi_2&:= \mathop{\mathrm{Law}} \Big(\sigma B + \int \mu^c_t(B)dt\, ,\, -\sigma B + \int \hat{\mu}^c_t(-B)dt \Big).
	\end{align*}
	These couplings share the same marginals and each of them is bi-causal. 
	It is easy to compute
	\begin{align*}
	\mathbb{E}_{\pi_1}\big[ \langle M\rangle_T^{p/2} + |A|_{\text{1-var}}^p\big]  &= (2c)^p,\\
	\mathbb{E}_{\pi_2}\big[ \langle M\rangle_T^{p/2} + |A|_{\text{1-var}}^p\big]  &=(8\sigma^2)^{p/2}.
	\end{align*}
	We conclude that, for each $p$, there are plenty of pairs $(c,\sigma)$ such that the ``synchronous'' coupling $\pi_1$ is not optimal between its marginals for the metric $\mathcal{AW}_p$.
\end{example}

{
To close this section,  we estimate the distance between two geometric Brownian motions with different volatilities.
\begin{proposition}\label{geoBM} 
For $i=1,2$, let $\mathbb{P}^{\sigma_i}$ be the law of the solution to the SDE $dZ^i_t=\sigma_i Z^i_t dB_t^i$ with $Z^i_0=1$, where $B^i$ denotes Brownian motion and $\sigma^i\in \R_+$. Letting $R\sim N(0,T)$, we then have
$$\mathcal{AW}_2(\mathbb{P}^{\sigma_1},\mathbb{P}^{\sigma_2} )^2=\mathbb E\left[\left(e^{\sigma_1R-\frac{\sigma_1^2T}{2}}-e^{\sigma_2 R-\frac{\sigma_2^2T}{2}} \right)^2\right]=e^{\sigma_1^2T}-2e^{\sigma_1\sigma_2T} +e^{\sigma_2^2T}$$
and for $p>1$
$$\mathcal{AW}_p(\mathbb{P}^{\sigma_1},\mathbb{P}^{\sigma_2} )^p\leq c_p \mathbb E\left[\left(e^{\sigma_1R-\frac{\sigma_1^2T}{2}}-e^{\sigma_2R-\frac{\sigma_2^2T}{2}} \right)^p\right],$$
where $c_p$ is the constant in the BDG-inequality which allows to control quadratic variation by  terminal value. 
\end{proposition}
\begin{proof} We have
	\begin{align*}
	& \mathcal{AW}_p(\mathbb{P}^{\sigma_1},\mathbb{P}^{\sigma_2} )^p \\ = & \inf\Big\{ \mathbb{E}_\pi\big[ \langle Z^1-Z^2\rangle_T^{p/2}\big]  : \pi \in \cplba(\P,\Q)\Big\}\\
	 \leq & c_p \inf\Big\{ \mathbb{E}_\pi[(Z_T^1-Z_T^2)^p]
	  : \pi \in \cplba(\P,\Q)\Big\} \\
	   = & c_p \inf\Big\{ \int \left(e^{\sigma_1 r_1 -\frac{\sigma_1^2T}{2}}-e^{\sigma_2r_2-\frac{\sigma_2^2T}{2}} \right)^p\, d\pi(r_1,r_2)	  : \pi \in \cpl(\gamma_T,\gamma_T)\Big\} \\
	   = & c_p \mathbb E\left[\left(e^{\sigma_1R-\frac{\sigma_1^2T}{2}}-e^{\sigma_2R-\frac{\sigma_2^2T}{2}} \right)^p\right],
	\end{align*}
	where $\gamma_T$  denotes a centered Gaussian with variance $T$.
%where the last inequality uses that the  optimizer is the given by synchronous coupling as in the proof of Proposition \ref{lem synchron}. 
For $p=2$ and $c_2=1$ we obtain equality.
% . The reason is that $\mathcal{AW}_2(\mathbb{P}^{\sigma_1},\mathbb{P}^{\sigma_2} )^2$ is essentially equivalent to the problem of finding a bivariate coupling with maximal correlation having fixed log-normal marginals, and such problem is attained by a co-monotone coupling. 
\end{proof}
}

{\subsection{Choice of the `cost functional'}\label{sec:ChoiceOfCost}
Recall from Definition \ref{AWD} 
that  the adapted Wasserstein distance is given through
	\begin{align*} 
	\mathcal{AW}_p(\mathbb{P},\mathbb{Q}):=\inf\{ \Phi  : \pi \in \cplba(\P,\Q)\Big\},
	\end{align*}
	where the `cost functional'  \begin{align}\label{CostFunctional} \Phi=\mathbb{E}_\pi\big[ \langle M^X-M^Y\rangle_T^{p/2} + |A^X-A^Y|_{\text{1-var}}^p\big]^{1/p}\end{align}
is defined using the semimartingale decompositions	
	$X=M^X+A^X, Y=M^Y+A^Y$.
The distinctive property of this ``quadratic plus first variation'' functional is that it exhibits the proper scaling to interpret the discrete time case as approximation to the continuous time counterpart. 
To wit, consider $\Omega=C([0,1])$, and let $\mathbb{P}^\sigma$ be the law of $X$ where $X_t=\int_0^t \sigma_s\, dB_s$, $B$ Brownian motion and  $\sigma\in C([0,1]), \sigma\geq 0$. 
	For each $N$, denote by $\mathbb{P}_N^\sigma$ the law of a random walk on $\{0,1/N,2/N, \dots,1\}$ with independent increments from $n/N$ to $(n+1)/N$ distributed according to $\mathcal N(0,\sigma_{n/N}^2/N)$.
	Then one can compute that for $0\leq \sigma, \sigma'\in C([0,1])$
	\begin{align*} 
	\mathcal{AW}_2(\mathbb{P}^\sigma_N,{\mathbb{P}}^{\sigma'}_N)
	&= \Big( \sum_{n=0}^{N-1} \frac{1}{N}|\sigma_{n/N}-\sigma_{n/N}'|^2 \Big)^{1/2} \\
	&\to \Big( \int_0^1 |\sigma_t-\sigma_t'|^2\,dt)\Big)^{1/2}
	=\mathcal{AW}_2(\mathbb{P}^\sigma,\mathbb{P}^{\sigma'} ).
	\end{align*}
For comparison, consider the consequences of replacing $\Phi$ in \eqref{CostFunctional} with $\tilde \Phi= \E_\pi[\sum_{i=0}^N (X_i-Y_i)_i^2]^{1/2}$ corresponding to quadratic nested distance (in terms of Pflug and Pichler \cite{PfPi12}).
While $\widetilde{\mathcal{AW}}_2$ and ${\mathcal{AW}}_2$ are equivalent metrics for each \emph{fixed} $N$, $\widetilde{\mathcal{AW}}_2$ does not exhibit the appropriate scaling for large $N$.
 A straightforward computation shows $\widetilde{\mathcal{AW}}_2(\mathbb{P}_N^\sigma,\mathbb{P}_N^{\sigma'})\to\infty$ as $N\to\infty$ whenever $\sigma\neq\sigma'$. In consequence, 
bounds on the hedging error  in terms of $\widetilde{\mathcal{AW}}_2(\mathbb{P}_N^\sigma,\mathbb{P}_N^{\sigma'})$ become progressively weaker as $N\to \infty$. In particular they do not allow for a meaningful continuous time limit.}
	
{When restricting solely to martingale measures $\P, \Q$, a sensible alternative to \eqref{CostFunctional} would be to consider the maximum norm, i.e.\ $\Phi'= \E_\pi[\sup_{t}|X_t-Y_t|^p]^{1/p}$. In fact,  
by  the BDG-inequalities this is essentially equivalent our choice in \eqref{CostFunctional}.
However, when considering semimartingales, this cost is too coarse. 
	For example, let $(\omega_n)$ be a sequence in $\Omega$ which converges to zero in maximum norm but for which the first variation tends to infinity.
	Then $\mathbb{P}_n:=\delta_{\omega_n}$ converges to $\mathbb{P}:=\delta_0$ (when adapted distance is defined only with maximum norm as cost), however, none of our optimization problems converge (take a strategy $H\in\mathcal{H}_k$ for which $(H(X)\bigcdot X)_T \approx k|\omega_n|_{\text{1-var}}$ almost surely).}

\subsection{Stochastic integrals and a contraction principle}

{We present here the two technical results which underlie the proofs of the main theorems in the article. The first one is}
\begin{lemma}
\label{lem:projection.of.strategy}
	Let $\mathbb{P},\mathbb{Q}\in\mathcal{SM}_1(\Omega)$, $H\in\mathcal{H}_k$, and  $\pi$ be a bi-causal coupling between $\mathbb{P}$ and $\mathbb{Q}$.
	Then there exists a process $G\in\mathcal{H}_k$ such that $G_t(Y)=\mathbb{E}_\pi[H_t(X)|Y]$ { for every $t$, $\pi$-almost surely.}
	{Moreover, we have $ (G(Y)\bigcdot Y)_T=\mathbb{E}_\pi[ (H(X)\bigcdot Y)_T|Y]$, $\pi$-almost surely.}
\end{lemma}
\begin{proof}
	In discrete time, write $H=\sum_{t=1}^{N} H_t 1_{\{t\}}$ for Borel functions $H_t\colon \mathbb{R}^{t-1}\to[-k,k]$.
	Let $\pi=\pi_\eta(d\omega)\,\mathbb{P}(d\omega)$ be a disintegration and define $$G'_t(\eta):=\int H_t(\omega)\pi_\eta(d\omega),$$ for every $t$ and $\eta\in\Omega$.
	By definition of bi-causal coupling  $G'_t$ is $\mathcal{F}^\Q_{t-1}$-measurable.
	It remains to pick functions $G_t$ which are $\mathcal{F}_{t-1}$ measurable such $G_t=G'_t$ $\mathbb{Q}$-almost surely.
	Since $\mathbb{E}_\pi[H_t(X)|Y]=G_t(Y)$ $\pi$-almost surely, it is clear that $ (G(Y)\bigcdot Y)_T=\mathbb{E}_\pi[ (H(X)\bigcdot Y)_T|Y]$ $\pi$-almost surely.

	In continuous time we
	{take $G$ to be the predictable projection of $H$, under the reference measure $\pi$, with respect to the $\pi$-completion of the filtration $\{\emptyset,\Omega\}\otimes \mathbb F^Y$. By \cite[Lemma C.1]{AcBaZa16} the result is $\pi$-indistinguishable from a predictable process under the $\mathbb{Q}$-completion of the filtration  $\mathbb F^Y$. The $t$-by-$t$, $\pi$-almost sure equality $G_t(Y)=\mathbb{E}_\pi[H_t(X)|Y]$, is then a consequence of the definition of predictable projection.} The $\pi$-almost sure equality $ (G(Y)\bigcdot Y)_T=\mathbb{E}_\pi[ (H\bigcdot Y)_T|Y]$ is established in Lemma \ref{lem technical proj} below, assuming that $\mathbb E_{\mathbb Q}[\langle Y \rangle_T]<\infty$. The general case follows by localization.
\end{proof}

\begin{lemma}
\label{lem technical proj} In the continuous-time context of Lemma \ref{lem:projection.of.strategy}, assume further that $\mathbb E_{\mathbb Q}[\langle Y \rangle_T]<\infty$. Then we have $ (G(Y)\bigcdot Y)_T=\mathbb{E}_\pi[ (H(X)\bigcdot Y)_T|Y]$, $\pi$-almost surely.
\end{lemma}

\begin{proof}
{The statement is true if instead of the stochastic integrals we considered the integrals w.r.t.\ the finite variation part of $Y$ (either by properties of Riemann-Stieltjes integrals, or directly from the definition of predictable projection). For this reason we may now assume that $Y$ is itself a martingale. }

We first take for granted the following result: if $h$ is bounded and predictable in the filtration of $(X,Y)$, and if $g$ denotes its predictable projection in the filtration of $Y$ under the measure $\pi$, then
\begin{align}\label{eq aux Jensen proj}
\mathbb E_\pi\Big[\int_0^T |g_t|^2d\langle Y\rangle_t \Big]\leq \mathbb E_\pi\Big[\int_0^T |h_t|^2d\langle Y\rangle_t \Big]. 
\end{align}
We know that there exist a sequence $(H^n)$ of predictable simple processes s.t. $$\lim_{n\to\infty}\mathbb E_\pi\Big[\int_0^T |H_t-H^n_t|^2d\langle Y\rangle_t \Big]=0.$$ 
By It\^o isometry the stochastic integrals $(H^n\bigcdot Y)_T$ converge in $L^2(\pi)$ to $(H\bigcdot Y)_T$. Denoting by $G^n$ the predictable projection of $H^n$ with respect to the $Y$-filtration, we deduce from \eqref{eq aux Jensen proj} that $$\lim_{n\to\infty}\mathbb E_\pi\Big[\int_0^T |G_t-G^n_t|^2d\langle Y\rangle_t \Big]=0,$$
so again by It\^o isometry $(G^n\bigcdot Y)_T$ converges in $L^2(\pi)$ to $(G\bigcdot Y)$. 
The $\pi$-almost sure equality $(G^n\bigcdot Y)_T=\mathbb{E}_\pi[(H^n\bigcdot Y)_T|Y]$ follows easily by the bi-causality of the coupling $\pi$, and by taking $L^2$ limits the desired conclusion is obtained.

To finish the proof we must establish \eqref{eq aux Jensen proj}. First we observe that
\begin{align*}\mathbb E_\pi\Big[\int_0^T |g_t|^2d\langle Y\rangle_t \Big]^{1/2}&=\sup_{\substack{f \text{ is $Y$-predictable}\\ \|f\|\leq 1} }\mathbb E_\pi\Big[\int_0^T f_tg_td\langle Y\rangle_t \Big]\\ &=\sup_{\substack{f \text{ is $Y$-predictable}\\ \|f\|\leq 1} }\mathbb E_\pi\Big[\int_0^T f_th_td\langle Y\rangle_t \Big],\end{align*}
as follows from predictable projection and upon taking $\|f\|^2:=\mathbb E_\pi[\int_0^1 |f_t|^2d\langle Y\rangle_t ]$. The result is a consequence of the equality
$$\mathbb E_\pi\Big[\int_0^T |h_t|^2d\langle Y\rangle_t \Big]^{1/2}=\sup_{\substack{f \text{ is $(X,Y)$-predictable}\\ \|f\|\leq 1} }\mathbb E_\pi\Big[\int_0^T f_th_td\langle Y\rangle_t \Big].$$
\end{proof}

Our next crucial technical result is given in Theorem \ref{thm:contraction} below. But first we need some preparation.

\begin{lemma}
\label{lem:estimating.integrals}
	Let $\mathbb{P},\mathbb{Q}\in\mathcal{SM}_p(\Omega)$, let $\pi$ be a bi-causal coupling between $\mathbb{P}$ and $\mathbb{Q}$, let $H\in\mathcal{H}_k$, and write $X-Y=M+A$ for the semimartingale decomposition under $\pi$.
	Then, for every $p\geq 1$, we have
	\begin{align*}
		\mathbb{E}_\pi[\|X-Y\|_\infty^p]
		&\leq  2^{p-1}b_p \cdot \mathbb{E}_\pi[\langle M \rangle_T^{p/2}+|A|_{\text{1-var}}^p],\\
		\mathbb{E}_\pi[|(H(X)\bigcdot X)_T - (H(X)\bigcdot Y)_T|^p]
		&\leq 2^{p-1}b_pk^p \cdot \mathbb{E}_\pi[\langle M\rangle_T^{p/2}+|A|_{\text{1-var}}^p],
	\end{align*}
	where $b_p$ is the upper constant in the BDG-inequality.
	If further $H_t\colon\Omega\to\mathbb{R}$ is $\tilde{L}$-Lipschitz continuous for every $t$, then we have
	\begin{align*}
	\mathbb{E}_\pi[|(H(X)\bigcdot X)_T - (H(Y)\bigcdot Y)_T|^p]
		&\leq 2^{2p-2}b_pk^p \cdot \mathbb{E}_\pi[\langle M\rangle_T^{p/2}+|A|_{\text{1-var}}^p]  \\
		&\quad+ \alpha \cdot \mathbb{E}_\pi[\langle M \rangle_T^{p}+|A|_{\text{1-var}}^{2p}]^{1/2}
	\end{align*}
	where $\alpha=2^{3p-2}\tilde{L}^pb_p b_{2p}^{1/2}\min\{\mathcal{AW}_{2p}(\mathbb{P},\delta_0)^p,\mathcal{AW}_{2p}(\mathbb{Q},\delta_0)^p\}$. 
\end{lemma}
\begin{proof}
	The elementary inequality $(x+y)^p\leq 2^{p-1} x^p+ 2^{p-1} y^p$ for $x,y\geq 0$ together with BDG inequality and the fact that $\|\cdot\|_\infty\leq |\cdot|_{\text{1-var}}$ imply
	\begin{align*}
	\mathbb{E}_\pi[\|X-Y\|_\infty^p]
	&\leq 2^{p-1} \mathbb{E}_\pi[\|M\|_\infty^p] + 2^{p-1}\mathbb{E}_\pi[|A|_{\text{1-var}}^p]\\
	&\leq 2^{p-1} b_p  \mathbb{E}_\pi[\langle M \rangle_T^{p/2}+|A|_{\text{1-var}}^p].
	\end{align*}
	This proves the first part. 
	The same arguments imply 	
	\begin{align*}
	&\mathbb{E}_\pi[|(H(X)\bigcdot X)_T - (H(X)\bigcdot Y)_T|^p]\\
	&\leq 2^{p-1} \mathbb{E}_\pi[| (H(X)\bigcdot M)_T|^p] + 2^{p-1} \mathbb{E}_\pi[| (H(X)\bigcdot A)_T|^p]\\
	&\leq 2^{p-1} k^p b_p \mathbb{E}_\pi[\langle M\rangle_T^{p/2}+|A|_{\text{1-var}}^p]
	\end{align*}
	from which the second part follows.
	To prove the third claim, write
	\begin{align*}
	&\mathbb{E}_\pi[| (H(X)\bigcdot X)_T - (H(Y)\bigcdot Y)_T|^p]\\
	&\leq 2^{p-1} \mathbb{E}_\pi[|((H(X)-H(Y))\bigcdot X)_T|^p] 
	+2^{p-1} \mathbb{E}_\pi[| (H(Y)\bigcdot X)_T - (H(Y)\bigcdot Y)_T|^p].
	\end{align*}
	The second term is smaller than $2^{p-1}2^{p-1}k^p b_p  \mathbb{E}_\pi[\langle M\rangle_T^{p/2}+|A|_{\text{1-var}}^p]$ by the second part.
	It remains to estimate $\mathbb{E}_\pi[|((H(Y)-H(Y))\bigcdot X)_T|^p]$.
	Write $X=N+B$ for the semimartingale decomposition of $X$ under $\mathbb{P}$. 
	By Lemma \ref{lem:semimartingales.remain.semimartingale}, the semimartingale decomposition under $\pi$ is still $X=N+B$.
	Moreover, the BDG-inequality, the Lipschitz-continuity of $H$, and H\"older's inequality, imply that
	\begin{align*}
	&\mathbb{E}_\pi[|((H(X)-H(Y))\bigcdot X)_T|^p]\\
	&\leq 2^{p-1} \mathbb{E}_\pi[|((H(X)-H(Y))\bigcdot N)_T|^p + |((H(X)-H(Y))\bigcdot B)_T|^p]\\
	&\leq 2^{p-1} \mathbb{E}_\pi[\|H(X)-H(Y)\|_\infty^p ( b_p \langle N \rangle_T^{p/2} + |B|_{\text{1-var}}^p)]\\
	&\leq 2^{p-1} b_p \tilde{L}^p \mathbb{E}_\pi[\|X-Y\|_\infty^{2p}]^{1/2} \mathbb{E}_\pi[(\langle N\rangle_T^p + |B|_{\text{1-var}}^P)^{2}]^{1/2}.
	% &=L'C_1 \mathcal{AW}_2(P,Q)\mathcal{AW}_2(P,\delta_0).
	\end{align*}
	It now follows from the first part that
	\[\mathbb{E}_\pi[\|X-Y\|_\infty^{2p}]^{1/2} \leq (2^{2p-1} b_{2p})^{1/2} \mathbb{E}_\pi[\langle M \rangle_T^{p}+|A|_{\text{1-var}}^{2p}]^{1/2}\]
	and by Lemma \ref{lem:semimartingales.remain.semimartingale} we have 
	\[\mathbb{E}_\pi[(\langle N\rangle_T^{p/2} + |B|_{\text{1-var}}^p)^{2}]^{1/2}\leq 2^{1/2}\mathcal{AW}_{2p}(\mathbb{P},\delta_0)^p.\]
	Putting all estimates together and replacing $X$ and $Y$ yields the claim.
\end{proof}

Denote by $\mathcal{P}_p(\mathbb{R})$ the set of all Borel probability measures $\mu$ on $\mathbb{R}$ such that $\int |x|^p\,\mu(dx)<\infty$.
Moreover, let $d_p(\mu,\nu)$ be the usual $p$-Wasserstein distance, and let $d_p^{w}$ the weak $p$-Wasserstein cost, that is,
\begin{align*} 
d_p(\mu,\nu)&:=\inf\Big\{ \Big(\int |x-y|^p\,\gamma(dx,dy)\Big)^{1/p} : \gamma\text{ is a coupling of }\mu\text{ and }\nu  \Big\},\\
d_p^{w}(\mu,\nu)&:=\inf\Big\{ \Big(\int \Big|x-\int y\,\gamma^x(dy)\Big|^p\,\mu(dx)\Big)^{1/p} : \gamma\text{ is a coupling of }\mu\text{ and }\nu  \Big\}.
\end{align*}
Here $\gamma=\mu(dx)\gamma^x(dy)$ denotes the disintegration.
Note that $d_p^{w}$ is not symmetric and as a consequence of Jensen's inequality, we always have $d_p^{w}\leq d_p$. Problems akin to $d_p^{w}(\mu,\nu)$ go under the name of `weak optimal transport' and have been recently introduced by Gozlan et al.\ in \cite{GoRoSaTe17}, but see also \cite{AlCoJo17,AlBoCh18,BaBePa18,BaBeHuKa17,GoRoSaSh18}.
We have
\begin{theorem}[Contraction]
\label{thm:contraction}
	Let $\mathbb{P},\mathbb{Q}\in\mathcal{SM}_p(\Omega)$, let $\pi$ a bi-causal coupling between $\mathbb{P}$ and $\mathbb{Q}$, let $C\colon\Omega\to\mathbb{R}$ be Lipschitz with constant $L$, and let $H\in\mathcal{H}_k$.
	Further denote by $X-Y=M+A$ the semimartingale decomposition under $\pi$ and let $G\in\mathcal{H}_k$ such that $(G(Y)\bigcdot Y)_T=\mathbb{E}_\pi[(H(X)\bigcdot Y)_T|Y]$ $\pi$-almost surely. 
	Then 
	\begin{align}\label{eq weak Wass}\begin{split}
	&d_p^{w}\Big( (C(Y)+(G(Y)\bigcdot Y)_T )(\mathbb{Q}) \, ,\, ( C(X)+(H(X)\bigcdot X)_T ) (\mathbb{P})  \Big)\\
	\leq & 2^{(p-1)/p}b_p^{1/p}(k+L) \cdot  \mathbb{E}_\pi[\langle M\rangle_T^{p/2}+|A|_{\text{1-var}}^p]^{1/p}.
	\end{split}
	\end{align}
	
	Now assume in addition that $H_t\colon\Omega\to\mathbb{R}$ is $\tilde{L}$-Lipschitz continuous for every $t$, then 
	\begin{align*} 
	&d_p\Big(  ( C(Y)+(H(Y)\bigcdot Y)_T ) ( \mathbb{Q} ) \,, \,( C(X)+(H(X)\bigcdot X)_T ) ( \mathbb{P} )\Big)\\
	\leq & 2^{(3p-3)/p}b_p^{1/p}(k+L) \mathbb{E}_\pi[\langle M\rangle_T^{p/2}+|A|_{\text{1-var}}^p]^{1/p}
	+\alpha^{1/p} \mathbb{E}_\pi[\langle M\rangle_T^{p}+|A|_{\text{1-var}}^{2p}]^{1/2p},
	\end{align*}
	where $\alpha$ is the constant of Lemma \ref{lem:estimating.integrals}.
\end{theorem}
\begin{proof}
	We start by proving the first claim.
	Let $\pi$ be as stated, and define $a(X):=C(X)+(H(X)\bigcdot X)_T$ as well as $b(Y):=C(Y)+ (G(Y)\bigcdot Y)_T$.
	Now let $\gamma:=(b(Y),a(X))(\pi)$ so that $\gamma$ is trivially a coupling between $b(Y)(\mathbb{Q})$ and $a(X)(\mathbb{P})$. 
	Therefore
	\begin{align*}
	d_p^{w}\Big( b(Y)(\mathbb{Q}) \, ,\, a(X) ( \mathbb{P} ) \Big)
	&\leq  \mathbb{E}_\pi[ \,|\,b(Y) - \mathbb{E}_\pi[ a(X)|b(Y)]\,|^p\,]^{1/p}.
	\end{align*}
	By assumption it holds that
	\[ \mathbb{E}_\pi[(G(Y)\bigcdot Y)_T- (H(X)\bigcdot X)_T |Y]
	= \mathbb{E}_\pi[ (H(X)\bigcdot Y)_T -(H(X)\bigcdot X)_T |Y]. \]
	Thus, using the tower property and Jensen's inequality, it follows that
	\begin{align*}
	&\mathbb{E}_\pi[ |b(Y) - \mathbb{E}_\pi[ a(X)|b(Y)]|^p]^{1/p}\\
		&\leq \mathbb{E}_\pi\big[\, \big| \,\mathbb{E}_\pi[C(Y)-C(X)|Y] + \mathbb{E}_\pi[(G(Y)\bigcdot Y)_T- (H(X)\bigcdot X)_T |Y] \,\big|^p\,\big]^{1/p}\\ 
	&\leq \mathbb{E}_\pi[ | C(Y)-C(X)|^p]^{1/p}+ \mathbb{E}_\pi[ \,|(H(X)\bigcdot Y)_T- (H(X)\bigcdot X)_T\,|^p\,]^{1/p}
	\end{align*}
	The claim now follows from the first and second estimates in Lemma \ref{lem:estimating.integrals}.

	In the second case where $H$ is additionally Lipschitz, let $d(X):=C(X)+(H(X)\bigcdot X)_T$ as well as $e(Y):=C(Y)+ (H(Y)\bigcdot Y)_T$ and $\gamma:=(e(Y),d(Y))(\pi)$.
	Then, similarly as before, 
	\begin{align*}
	&d_p\Big( e(Y)(\mathbb{Q}) \, ,\, d(X) ( \mathbb{P} ) \Big)
	\leq \mathbb{E}_\pi[ |e(Y) - d(Y)|^p]^{1/p} \\
	&\leq \mathbb{E}_\pi[ |C(Y) - C(X)|^p]^{1/p} + \mathbb{E}_\pi[ |(H(Y)\bigcdot Y)_T - (H(X)\bigcdot X)_T|^p]^{1/p}
	\end{align*}
	and the claim follows from the first and third estimates of Lemma \ref{lem:estimating.integrals}.
\end{proof}

\begin{remark}
	An evident question is whether an estimate for the usual Wasserstein distance holds true without the (Lipschitz-) continuity assumption on $H$. 
	Namely if \eqref{eq weak Wass} holds for $d_p$ instead of $d_p^w$.
	The following example shows that this is not true.
	In a two-period discrete time model $(T=2)$, let 
	\[\mathbb{P}:=\delta_{0}\otimes((\delta_1+\delta_{-1})/2) \quad\text{and}\quad 
	\mathbb{P}_\varepsilon:=((\delta_\varepsilon+\delta_{-\varepsilon})/2)\otimes((\delta_1+\delta_{-1})/2)\]
	so that $\mathcal{AW}_p(\mathbb{P}_\varepsilon,\mathbb{P})\to 0$ as $\varepsilon\to 0$ for every $p$.
	Then, set $H_1:=0$ and $H_2:=1_{(0,\infty)}-1_{(-\infty,0)}$.
	For the projection under any bi-causal coupling between $\mathbb{P}_\varepsilon$ and $\mathbb{P}$ of $H$ onto $Y$ one computes $G_1=0$ and $G_2=0$.
	In particular $(G(Y)\bigcdot Y)_T=0$ $\mathbb{P}$-almost surely.
	However, for every $\varepsilon>0$ one has $\mathbb{P}_\varepsilon( (H(X)\bigcdot X)_T\geq 1-\varepsilon )\geq 1/4$ which implies that the respective laws cannot converge.
\end{remark}

\begin{remark}
\label{BDGremark} 
By $b_p$ we denote
 the smallest real number such that 
\begin{align}\E[ \|M\|^p_\infty] \leq b_p \E[ [M]^{p/2}]  
\end{align}
 for every martingale $M$. For $p\geq 2$ it was established by Burkholder \cite{Bu91} that $b_p=p$  but the value of $b_p$ is unknown for $p\in [1,2)$ according to  \cite{Os10}, \cite[page 427]{Os12}. By \cite{BeSi13}, $b_1 \leq 6$. (The optimal constant in the reverse inequality is known for the trivial case $p=2$ and for $p=1$. In the latter instance one obtains $\sqrt 3$ 
\cite{Bu02} and $1.2727\ldots$ \cite{ScSt15} for continuous martingales, resp.)

\end{remark}

\section{Proofs of the results stated in the introduction and extensions}\label{sec proofs intro}

Thanks to work done in the previous section, the strategy for the proofs boils down into two parts. In a first step, one forgets about the space $\Omega$ and only focuses on continuity of the problem at hand with respect to $d_p$ or $d_p^{w}$ when image measures on $\mathbb{R}$ are plugged in: e.g.~in utility maximization this means to study continuity of $\mu\mapsto \int U(x)\,\mu(dx)$.
In a second step, one uses the obtained continuity and the contraction theorem in the previous section.

\subsection{Proof of Theorem \ref{thm:hedging.with.loss.linear}}
We will need the elementary estimate
\begin{lemma}
\label{lem:trivial.lemma}
	Let $\mu,\nu\in\mathcal{P}_1(\mathbb{R})$ and let $f\colon\mathbb{R}\to\mathbb{R}$ be convex and Lipschitz.
	\begin{align}\label{ineq lip conv}
		\int f(x)\mu(dx)-\int f(y)\,\nu(dy)\leq L\, d^{w}_1(\mu,\nu),
		\end{align}
	where $L$ is Lipschitz constant of $f$.
\end{lemma}
\begin{proof}
	Let $\gamma$ be a coupling of $\mu$ and $\nu$.
	Applying Jensen's inequality we obtain
	\begin{align*}
	&\int f(x)\,\mu(dx)-\int f(y)\,\nu(dy)
	=\int f(x)-f(y)\,\gamma(dx,dy)\\
	=&\int \Big  (f(x) -\int f(y)\,\gamma^x(dy)\Big )\,\mu(dx)
	\leq \int \Big ( f(x) - f\Big (\int y \,\gamma^x(dy)\Big )\Big)\,\mu(dx)\\
	\leq & L\int \Big| x- \int y\,\gamma^x(dy)\Big|\,\mu(dx).
	\end{align*}
	As $\gamma$ was arbitrary, this implies the claim.
\end{proof}
In fact there is equality in the previous lemma, if one takes supremum in the l.h.s.\ of \eqref{ineq lip conv} over all $L$-Lipschitz convex function, as shown in \cite[Proposition 3.2]{GoRoSaTe17}.

{	We now turn to the proof of Theorem \ref{thm:hedging.with.loss.linear}.	
	{For $n>0$ let $\pi$ be } a bi-causal coupling which attains the infimum in the definition of $\mathcal{AW}_{1}(\mathbb{P},\mathbb{Q})$ modulo a $1/n$-margin.
	By Lemma \ref{lem:projection.of.strategy} there is $G^n\in\mathcal{H}_k$ such that $(G^n(Y)\bigcdot Y)_T=\mathbb{E}_\pi[(H(X)\bigcdot Y)_T|Y]$ $\pi$-almost surely. 
	%In fact 	let $\gamma$ be a $1/n$-optimizer for $d_1^w( \mu^n,\nu)$, where 
	Define
	\[\mu^n:=(C(Y)+(G^n(Y)\bigcdot Y)_T)(\mathbb{Q})\quad\text{and}\quad \nu:=(C(X)+(H(X)\bigcdot X)_T)(\mathbb{P}).\]
	(Note that $\mu^n,\nu\in\mathcal{P}_1(\mathbb{R})$ as $\mathbb{P},\mathbb{Q}\in\mathcal{SM}_1(\Omega)$.)
	By Lemma \ref{lem:trivial.lemma} we have%Then it follows that
	\begin{align*}
	&\mathbb{E}_\mathbb{Q}\big[ (C(Y)- m - (G^n(Y)\bigcdot Y)_T)^+ \big] -	\mathbb{E}_\mathbb{P}\big [(C(X) - m -(H(X)\bigcdot X)_T)^+ \big]\\ 	\leq &  d^{w}_1(\mu^n,\nu).
	\end{align*}
	From Theorem \ref{thm:contraction} we obtain
		\begin{align}\label{eq almost superhedge n}
	\begin{split}
	\mathbb{E}_\mathbb{Q}[(C(Y)-m-(G^n(Y) \bigcdot Y)_T))^+]
	&\leq \mathbb{E}_\mathbb{P}[(C(X)-m-(H(X)\bigcdot X)_T)^+]\\
	&\quad+b_1(k+L) \,(\mathcal{AW}_1(\mathbb{P},\mathbb{Q})+1/n).
	\end{split}
	\end{align}

	Assume first that $\mathbb{E}_\mathbb{Q}[[Y]_T]<\infty$ and denote by $A$ the finite variation process associated to $Y$.
	Then, as $(G^n)$ is uniformly bounded by $k$, there exists a predictable $G$ and a sequence of forward-convex combinations of $(G^n)$ which converge in $L^2(d\mathbb{Q}\otimes d([Y]+A))$ to $G$. 
	This, \eqref{eq almost superhedge n}, and the convexity of $(\cdot)^+$ lead to the desired conclusion.}
	{The general case follows by a simple but notationally heavy localization argument.}
 
	The proof in case that $G=H$ and $H$ is Lipschitz follows analogously from the second part of Theorem \ref{thm:contraction}.

\subsection{Proof of Theorem \ref{thm:avar}}

	In a first step notice that for all $\mathbb{P},\mathbb{P}'$ and random variables $Z,Z'$, it follows as in Lemma \ref{lem:trivial.lemma} that 
	\[\mathrm{AVaR}_\alpha^\mathbb{P}(Z)-\mathrm{AVaR}_\alpha^{\mathbb{P}'}(Z')\leq d_1^{w}(Z(\mathbb{P}),Z'(\mathbb{P}'))/\alpha.\]
{Indeed, if $\gamma$ is a coupling from $\mu:=Z(\mathbb P)$ to $\nu:=Z'(\mathbb P')$ then 
	\begin{align*}
	&\mathrm{AVaR}_\alpha^\mathbb{P}(Z)-\mathrm{AVaR}_\alpha^{\mathbb{P}'}(Z')\\ %= & \inf_m\int \Big [\frac{1}{\alpha}[x-m]^+-m  \Big ]\gamma(dx,\mathbb R)- \inf_m\int \Big [\frac{1}{\alpha}[y-m]^+-m  \Big ]\gamma^x(dy)\gamma(dx,\mathbb R)\\
	= &
	\inf_m\int  \frac{1}{\alpha}(x-m)^+-m  \,\mu(dx) - \inf_m \frac{1}{\alpha}\int\int (y-m)^+\gamma^x(dy)-m \,\mu(dy)\\
	\leq &
	\sup_m  \frac{1}{\alpha}\int (x-m)^+-(y-m)^+ \,\gamma(dx,dy)\\
	\leq &
	\sup_m \frac{1}{\alpha}\int (x-m)^+-\Big(\int y\,\gamma^x(dy)-m\Big)^+ \,\mu(dx) \\
	\leq & \frac{1}{\alpha}\int \Big |x- \int y \, \gamma^x(dy) \Big |\,\mu(dx),
	\end{align*}
	so minimizing over $\gamma$ yields the claim.}

	The rest of the proof now follows the line of argumentation as in the proof for Theorem \ref{thm:hedging.with.loss.linear}. 
	Fix $\mathbb{P},\mathbb{Q}\in\mathcal{SM}_1(\Omega)$. 
	Assume only for notational simplicity that there exists a bi-causal coupling $\pi$ which attains the infimum in the definition of $\mathcal{AW}_{1}(\mathbb{P},\mathbb{Q})$, and that there exist	$H^\ast\in\mathcal{H}_k$ such that 
	\[\mathrm{AVaR}_\alpha^\mathbb{P}(C(X)-(H^\ast(X)\bigcdot X)_T)
	=\inf_{H\in\mathcal{H}_k}\mathrm{AVaR}_\alpha^\mathbb{P}(C(X)-(H(X)\bigcdot X)_T).\]
	By Lemma \ref{lem:projection.of.strategy} there is $G^\ast\in\mathcal{H}_k$ such that $(G^\ast(Y)\bigcdot Y)_T=\mathbb{E}_\pi[(H^\ast(X)\bigcdot Y)_T|Y]$ $\pi$-almost surely. 
	Therefore   
	\begin{align*}  
	& \inf_{G\in\mathcal{H}_k}\mathrm{AVaR}_\alpha^\mathbb{Q}(C(Y)-(G(Y)\bigcdot Y)_T) -\inf_{H\in\mathcal{H}_k}\mathrm{AVaR}_\alpha^\mathbb{P}(C(X)-(H(X)\bigcdot X)_T)  \\
	&\leq \mathrm{AVaR}_\alpha^\mathbb{Q}(C(Y)-(G^\ast(Y)\bigcdot Y)_T) -\mathrm{AVaR}_\alpha^\mathbb{P}(C(X)-(H^\ast(X)\bigcdot X)_T) \\
	&\leq \frac{1}{\alpha} d_1^{w}\Big( (C(Y)-(G^\ast(Y)\bigcdot Y)_T) ( \mathbb{Q} ) \, , \, (C(X)-(H^\ast(X)\bigcdot X)_T) ( \mathbb{P} ) \Big) \\
	&\leq \frac{ b_1(k+L)}{\alpha} \mathcal{AW}_1(\mathbb{P},\mathbb{Q}) ,
	\end{align*} 
	where the last inequality is due to Theorem \ref{thm:contraction}.
	Interchanging the role of $\mathbb{P}$ and $\mathbb{Q}$ yields the desired conclusion. The proof for the second estimate follows analogously.

\subsection{Proof of Example \ref{ex:hedge.BM}}
	First note that $\mathrm{AVaR}^{\mathbb{P}}_\alpha(Z)\geq \mathbb{E}_\mathbb{P}[Z]$ for every integrable random variable $Z$.
	Indeed, this follows from integrating the pointwise inequality $x=x+m-m\leq (x+m)^+/\alpha-m$.
	Therefore, as the Brownian stochastic integral has expectation zero, we conclude that $\inf_{H\in\mathcal{H}_k}\mathrm{AVaR}^{\mathbb{P}}_\alpha(C(X)- (H(X)\bigcdot X)_T)\geq \mathbb{E}_{\mathbb{P}}[C(X)]$.
	On the other hand, define 
	\[f(t,x):=\int c(x+y) \, N(0,\sigma^2(T-t))(dy) \quad\text{for } (t,x)\in[0,T]\times\mathbb{R},\]
	where $N(0,\sigma^2(T-t))$ stands for the normal distribution with mean 0 and variance $\sigma^2(T-t)$.
	Then $C(X)=f(T,X_T)$ and $\mathbb{E}_{\mathbb{P}}[f(t,X_t)|\mathcal{F}_s]= f(s,X_s)$ for every $0\leq s\leq t\leq T$.
	Thus, by It\^o's formula and fact that the martingale property implies that the finite variation part vanishes, one has $f(t,X_t)=f(0,0)+ (H^\ast(X)\cdot X)_T$ for the predictable trading strategy $H^\ast_t:=\partial_x f(t,X_t)$.
	As further $|H^\ast_t|\leq 1$ for every $t$ and $f(0,0)=\sigma/\sqrt{2\pi}$, one has 
	\[\inf_{H\in\mathcal{H}_1}\mathrm{AVaR}^{\mathbb{P}}_\alpha(C(X)- (H(X)\bigcdot X)_T)
	\leq \mathrm{AVaR}^\mathbb{P}_\alpha(C(X)-(H^\ast(X)\cdot X)_T)
	=\frac{\sigma}{\sqrt{2\pi}}.\]
	The proof now follows from the explicit formula for the adapted Wasserstein distance derived in Example \ref{ex det diff} and the fact that $\mathbb{E}_\mathbb{P}[C(X)]=\sigma/\sqrt{2\pi}$.

\subsection{Proof of Theorem \ref{thm:utility}}

	Recall that $U'(x)\leq c(1+|x|^{p-1})$ for all $x\in\mathbb{R}$ and some constant $c$.
	Let $\mathbb{P},\mathbb{Q}\in\mathcal{SM}_p(\Omega)$ be arbitrary and assume only for notational simplicity that there is $H^\ast\in\mathcal{H}_k$ such that \[\mathbb{E}_\mathbb{P}[U(C(X)+(H^\ast(X)\bigcdot X)_T)]=\sup_{H\in\mathcal{H}_k} \mathbb{E}_\mathbb{P}[U(C(X)+(H(X)\bigcdot X)_T)]\] 
	and that there is a bi-causal coupling $\pi$ coupling between $\mathbb{P}$ and $\mathbb{Q}$ which is optimal for $\mathcal{AW}_p(\mathbb{P},\mathbb{Q})$.
	By Lemma \ref{lem:projection.of.strategy} there is $G^\ast\in\mathcal{H}_k$ such that $(G^\ast(Y)\bigcdot Y)_T=\mathbb{E}_\pi[(H^\ast(X)\bigcdot Y)_T|Y]$ $\pi$-almost surely.
	Let $$\mu:=(C(Y)+(G^\ast(Y)\bigcdot Y)_T)(\mathbb{Q})\text{  and  }\nu:= (C(X)+(H^\ast(X)\bigcdot X)_T)(\mathbb{P}),$$ and let $\gamma$ be an (almost) optimal coupling for $d_p^{w}(\mu,\nu)$.
	As $U$ is concave and increasing, we have $U(y)-U(x)\leq U'(\min\{x,y\})|x-y|$. 
	Using Jensen's inequality for the concave function $U$ we have
	\begin{align*}
	&\sup_{H\in\mathcal{H}_k} \mathbb{E}_\mathbb{P}[U(C(X)+(H(X)\bigcdot X)_T)] - \sup_{G\in\mathcal{H}_k} \mathbb{E}_\mathbb{Q}[U(C(Y)+(G(Y)\bigcdot Y)_T)]\\
	&\leq \mathbb{E}_\mathbb{P}[U(C(X)+(H^\ast(X)\bigcdot X)_T)] - \mathbb{E}_\mathbb{Q}[U(C(Y)+(G^\ast(Y)\bigcdot Y)_T)] \\
	&=\int U(y)-U(x)\,\gamma(dx,dy)
	\leq \int U\Big(\int y\,\gamma^x(dy)\Big) -U(x)\,\mu(dx)\\
	&\leq \Big( \int  \Big|U'\Big(\min\Big\{x,\int y\,\gamma^x(dy)\Big\}\Big)\Big|^q \,\mu(dx)\Big)^{1/q} d_p^{w}(\mu,\nu),
	\end{align*}
	where we used H\"older's inequality in the last line and $q$ denotes the conjugate H\"older exponent of $p$ (that is, $1/p+1/q=1$).
	As $q(p-1)=p$, the growth assumption on $U'$ implies that $|U'(\min\{x,y\})|^q\leq c(1+|x|^p+|y|^p)$ for some (new) constant $c$.
	Then, by Lemma \ref{lem:estimating.integrals}, we have
	\begin{align*}
	&\int  \Big|U'\Big(\min\Big\{x,\int y\,\gamma^x(dy)\Big\}\Big)\Big|^q \,\mu(dx)\\
	&\leq c\Big( 1+ \int |x|^p \,\mu(dx)+ \int \Big|\int y\,\gamma^x(dy)\,\Big|^p\,\mu(dy)\Big)\\
	&\leq c\Big( 1+ \int |x|^p \,\mu(dx)+ \int |y|^p\,\nu(dy)\Big)\\
	&\leq \tilde{c} \big(1+ \mathcal{AW}_p(\mathbb{Q},\delta_0)^p + \mathcal{AW}_p(\mathbb{P},\delta_0)^p \big)
	=:e
	\end{align*}
	for $e:=\tilde{c}(1+R^p+R^p$).
	Exchanging the roles of $\mathbb{P}$ and $\mathbb{Q}$ and using Theorem \ref{thm:contraction} completes the proof.

\subsection{The proof of Theorem \ref{thm:utility.indifference}}

{
In a first step, we claim that $v(\mathbb{P})$ is uniformly bounded over all $\mathbb{P}$ with $\mathcal{AW}_p(\mathbb{P},\delta_0)\leq R$.
	Indeed, using the growth assumption on $U$, the fact that $U$ is stricly increasing, and the BDG-inequality to control the $p$-th moment of $(H\bigcdot X)_T$, it follows that there exist $a,A\in\mathbb{R}$ such that 
	\begin{align}
	\label{eq:utility.indiff.prelimiary}
	\inf U
	< a
	\leq \sup_{H\in\mathcal{H}_k}\mathbb{E}_\mathbb{P}[U((H\bigcdot X)_T)]
	\leq A
	<\sup U
	\end{align}
	for all $\mathbb{P}$ with $\mathcal{AW}_p(\mathbb{P},\delta_0)\leq R$.
	Now assume that there exists a sequence $\mathbb{P}_n$ with $\mathcal{AW}_p(\mathbb{P}_n,\delta_0)\leq R$ but $v(\mathbb{P}_n)\to\infty$.
	Then, using the  BDG-inequality once more, it follows that 
	\[\sup_{H\in\mathcal{H}_k} \mathbb{E}_{\mathbb{P}_n}[U(C-v(\mathbb{P}_n)+(H\bigcdot X)_T)]
	\to \inf U, \]
	a contradiction to \eqref{eq:utility.indiff.prelimiary}.
	The case $v(\mathbb{P}_n)\to-\infty$ is excluded analogously.
	}
	
{
	At this point, using the definition of $v(\mathbb{P})$, a twofold application of Theorem \ref{thm:utility} 
	yields
	\[ \Big| \sup_{H\in\mathcal{H}_k} \mathbb{E}_\mathbb{Q}[U(C-v(\mathbb{P})+(H\bigcdot X)_T)] -\sup_{H\in\mathcal{H}_k} \mathbb{E}_\mathbb{Q}[U( (H\bigcdot X)_T)]\Big|\leq 2K \cdot \mathcal{AW}_p(\mathbb{P},\mathbb{Q}).\]
	Indeed, while a direct application of the theorem would give a constant $K$ which depends on $v(\mathbb{P})$, an inspection of its proof shows that the constant $K$ depends only on the size of $v(\mathbb{P})$.
	By the first step this is bounded unifomly over $\mathbb{P}$ with $\mathcal{AW}_p(\mathbb{P},\delta_0)\leq R$.
	}
	
{
	Now let $\varepsilon>0$ and $H\in\mathcal{H}_k$ be arbitrary, and set $Y:=Y^H:=C-v(\mathbb{P})+(H\bigcdot X)_T$.
	Then, it follows that there is some constant $c>0$ (depending on $R\geq\mathcal{AW}_p(\mathbb{Q},\delta_0)$ and $U$) such that
	\[\mathbb{E}_\mathbb{Q}[U(Y+\varepsilon)]
	= \mathbb{E}_\mathbb{Q}[U(Y)] +\mathbb{E}_\mathbb{Q}\Big[ \int_{Y}^{Y+\varepsilon} U'(z)\,dz \Big] 
	\geq \mathbb{E}_\mathbb{Q}[U(Y)] +\varepsilon c.\]
	Indeed, this would follow directly if $Y$ were bounded by a fixed constant but readily extends to the present setting as $\mathbb{E}_\mathbb{Q}[|Y|^p]\leq C$ for some constant $C>0$, independent of $H$ and $\mathbb{Q}$ as long as $\mathcal{AW}_p(\mathbb{Q},\delta_0)\leq R$.
	In a similar manner $\mathbb{E}_\mathbb{Q}[U(Y-\varepsilon)]\leq \mathbb{E}_\mathbb{Q}[U(Y)] -\varepsilon c$.
	}
	
{
	Putting everything together reveals
	\begin{align*}
	&\sup_{H\in\mathcal{H}_k} \mathbb{E}_\mathbb{Q}[U(C-(v(\mathbb{P})+\varepsilon)+(H\bigcdot X)_T)]
	< \sup_{H\in\mathcal{H}_k} \mathbb{E}_\mathbb{Q}[U( (H\bigcdot X)_T)]\\
	&< \sup_{H\in\mathcal{H}_k} \mathbb{E}_\mathbb{Q}[U(C-(v(\mathbb{P})-\varepsilon)+(H\bigcdot X)_T)] 
	\end{align*}
	for some $\varepsilon< \tilde{C} \mathcal{AW}_p(\mathbb{P},\mathbb{Q})$ (where $\tilde{C}$ is a new constant emerging from $K$ and $c$).
	Thus $|v(\mathbb{Q})-v(\mathbb{P})|\leq \varepsilon\leq \tilde{C} \mathcal{AW}_p(\mathbb{P},\mathbb{Q})$ which completes the proof.
}
\subsection{Two generalizations}

The following two results can be proved using almost the same arguments as used in the proofs of Theorem \ref{thm:avar} and Theorem \ref{thm:utility}.
In particular the proofs boil down to establishing convergence for image measures with respect to $d_p$ and give no new insight on adapted Wasserstein distances, so we shall skip them.

\begin{proposition}
\label{prop:hedging.loss}
	Let $\ell\colon\mathbb{R}\to\mathbb{R}_+$ be a convex and strictly increasing function and let $\delta>0$.
	Assume that $p\geq 1$ is such that $\ell'(x)\leq c(1+|x|^{p-1})$ for some constant $c$. 
	Then, for every Lipschitz continuous function $C\colon\Omega\to\mathbb{R}$, the function
	\[\mathbb{P}\mapsto
	\inf\Big\{ m\in\mathbb{R} : \exists H\in\mathcal{H}_k \text{ such that }
		\mathbb{E}_\mathbb{P}[\ell(C(X)-(H(X)\bigcdot X)_T-m)]\leq\delta  \Big\} \] 
	is continuous on $(\mathcal{SM}_p(\Omega),\mathcal{AW}_p)$.
\end{proposition}

Let $\rho$ be a law-invariant risk measure which we directly view as a functional from $\mathcal{P}_p(\mathbb{R})$ to the reals.
For $\mathbb{P}\in\mathcal{SM}_p(\Omega)$ and a random variable $Z\colon\Omega\to\mathbb{R}$ (such that $Z( \mathbb{P})\in\mathcal{P}_p(\mathbb{R})$) we write $\rho^\mathbb{P}(Z)=\rho(Z(\mathbb{P}))$.
A typical example of a law invariant risk measure which satisfies $\rho(\mu)-\rho(\nu)\leq L d^w(\mu,\nu)$ for some constant $L$ depending on the $p$-the moment of $\mu$ and $\nu$ is the optimized certainty equivalent, introduced to the mathematical finance community in \cite{BeTe07}.
For a convex, increasing function $\ell\colon\mathbb{R}\to\mathbb{R}$ which is bounded from below and satisfies $\ell(x)/x\to\infty$ as $x\to\infty$, the optimized certainty equivalent is defined via
\[ \rho^\mathbb{P}(Z)
:=\inf_{m\in\mathbb{R}} \big( \mathbb{E}_\mathbb{P}[\ell(Z-m)]+m \big)
=\inf_{m\in\mathbb{R}} \Big( \int \ell(x-m)\,(Z( \mathbb{P}))(dx)+m \Big). \]
If $\ell'(x)\leq c(1+|x|^{p-1})$, then it follows that the infimum over $m$ can be taken in some compact set depending on the $p$-th moments.
Due to cash additivity of $\rho$, the following proposition has the same interpretation as Theorem \ref{thm:avar}.

\begin{proposition}
\label{prop:heding.law.invariant}
	Assume that $\rho\colon \mathcal{P}_p(\mathbb{R})\to\mathbb{R}$ satisfies $\rho(\mu)-\rho(\nu)\leq L d^w(\mu,\nu)$ for some constant $L$ depending on the $p$-the moment of $\mu$ and $\nu$.
	Then, for every Lipschitz function $C\colon\Omega\to\mathbb{R}$, the mapping
	\[ \mathbb{P} \mapsto \inf_{H\in\mathcal{H}_k} \rho^\mathbb{P}( C(X) - (H(X)\bigcdot X)_T ) \]
	is locally Lipschitz continuous on $(\mathcal{SM}_p(\Omega), \mathcal{AW}_p)$.
\end{proposition}

Finally, let us point out that (though not a convex risk measure) the Value-at-Risk ($\mathrm{VaR}$) would be another natural candidate to study continuity.
However, as $\mathrm{VaR}$ is not continuous w.r.t.~weak convergence, already in a one period model continuity of 
$\mathbb{P}\mapsto \inf\{ m\in\mathbb{R} : \text{there is } H\in\mathcal{H}_k \text{ with } \mathrm{VaR}^\mathbb{P}(C(X)-m-(H(X)\bigcdot X)_T)\leq 0\}$ does not hold.

\section{Final remarks}\label{sec the end}

\begin{remark}[Usual Wasserstein does not work I]
\label{ex:wasserstein.does.not.work}
	We note that convergence in the usual Wasserstein distance is not sufficient to obtain continuity in any of the problems we study in this paper. 
	Consider a two period market with 
	\begin{align*}
	\mathbb{P}_n&=\frac{1}{4}\Big( \delta_{(1/n,1)} +\delta_{(1/n,0)} + \delta_{(-1/n,0)} + \delta_{(-1/n,-1)}\Big),\\
	\mathbb{P}&=\frac{1}{4}\Big( \delta_{(0,1)} +2\delta_{(0,0)} +\delta_{(0,-1)}\Big).
	\end{align*}
	Then $\mathbb{P}$ and each $\mathbb{P}_n$ satisfy the classical no-arbitrage condition, unlike the situation described in Figure \ref{fig:usual.wasserstein}.
	While $\mathbb{P}_n$ converges to $\mathbb{P}$ in usual Wasserstein distance, one can verify that convergence in nested distance does not hold.
	For example in utility maximization of the trivial claim $C=0$, we have $\sup_{H\in\mathcal{H}_k} \mathbb{E}_\mathbb{P}[U(C(X)+(H(X)\bigcdot X)_T)]=U(0)$ by Jensen's inequality (as $X$ is a martingale under $\mathbb{P}$).
	For $\mathbb{P}_n$  taking the strategy $H^\ast$ consisting of $H^\ast_0=0$ and $H^\ast_1(x)=k\mathop{\mathrm{sign}}(x)$, one gets
	\[ \sup_{H\in\mathcal{H}_k} \mathbb{E}_{\mathbb{P}_n}[U(C(X)+(H(X)\bigcdot X)_T)]
	\geq \mathbb{E}_{\mathbb{P}_n}[U(C(X)+(H^\ast(X)\bigcdot X)_T)]
	\to U(k), \] 
	showing the lack of continuity.
\end{remark}

\begin{remark}[Usual Wasserstein does not work II]
\label{rem:superhedging.does.not.work}
	As explained in the introduction, the objective in Theorem \ref{thm:hedging.with.loss.linear} can be seen as a relaxed version of the superhedging problem.
	The reason to consider this relaxation is not a technical simplification but necessary to to obtain continuity without further assumptions.
	Indeed, the problem of superhedging 
	\[ \inf\Big\{ m\in\mathbb{R} : \text{there is }  H\in\mathcal{H}_k \text{ such that } m+ (H\bigcdot X)_T \geq C(X),  \,\mathbb{P}\text{-almost surely}\big\} \]
	is not continuous in $\mathbb{P}$ w.r.t.~adapted distance for any $k\in[0,\infty]$.
	In fact, this already happens in one period, where adapted and the usual Wasserstein distances coincide.
	Consider a sequence of measures $\mathbb{P}_n$ with full support which converge weakly to a measure $\mathbb{P}$.
	Then the superhedging price w.r.t.~$\mathbb{P}_n$ equals the concave envelope of $C$, while the superhedging price w.r.t.~$\mathbb{P}$ equals the concave envelope of	$C$ restricted to the support of $\mathbb{P}$.
	For a recent paper on this problem in one period, see the work of Ob{\l}{\'o}j and Wiesel \cite{ObWi18}.
\end{remark}

\begin{remark}[Uniformly bounded strategies are necessary]
\label{rem:bounded.strategies}
	Similar as in Remark  \ref{rem:superhedging.does.not.work} the restriction to trading strategies in $\mathcal{H}_k$ (i.e.~ uniformly bounded strategies) is also no technical simplification. 
	For example, in a one-period framework, the measures $\mathbb{P}^\varepsilon:=(1-\varepsilon)\delta_{(0,\varepsilon)}+\varepsilon\delta_{(0,-\varepsilon)}$ converges to $\mathbb{P}:=\delta_{(0,0)}$ in every (adapted) Wasserstein distance. 
	However, we have for small $\epsilon>0$
	\[ \inf_{H\in\mathcal{H}_\infty } \mathrm{AVaR}_\alpha^{\mathbb{P}^\varepsilon}((H\bigcdot X)_T)
	=-\infty
	\quad\text{while}\,\,\,
	\inf_{H\in\mathcal{H}_\infty } \mathrm{AVaR}_\alpha^{\mathbb{P}}((H\bigcdot X)_T)=0,\]
	where $H\in\mathcal{H}_\infty:=\bigcup_{k\in\mathbb{N}}\mathcal{H}_k$ is the set of all bounded trading strategies.
\end{remark}

 \medskip

\paragraph{\textbf{Acknowledgements}}
All authors are  grateful to the anonymous referees whose insightful comments had a significant impact on this article. J.\ Backhoff gratefully acknowledges financial support by the FWF through grant P30750 and by the Vienna University of Technology. D.\ Bartl has been funded by the Austrian Science Fund (FWF) under Project P28661.
M.\ Beiglboeck and M.\ Eder gratefully acknowledge financial support by the FWF through grant Y782.


%%%%%%%%\bibliographystyle{spmpsci}      % mathematics and physical sciences
%\bibliographystyle{spmpsci}

\bibliographystyle{abbrv}
\bibliography{joint_biblio}

\end{document}